\documentclass[jphysa,nofootinbib]{iopart}
\usepackage{iopams}

\makeatletter
\@namedef{ver@amsmath.sty}{}
\makeatother
\usepackage{amstext}

\usepackage[latin1]{inputenc}
\usepackage{caption}
\usepackage{subcaption}
\usepackage{todonotes}
\usepackage{amsmathred}
\usepackage{amssymb}
\usepackage{bbold}
\usepackage{bbm}
\usepackage[pdftex]{hyperref}
\usepackage{braket}
\usepackage{dsfont}
\usepackage{mathdots}
\usepackage{mathtools}
\usepackage{enumerate}
\usepackage[shortlabels]{enumitem}
\usepackage{csquotes}
\usepackage{stmaryrd}
\usepackage[cal=boondox]{mathalfa}
\usepackage{graphicx}
\usepackage{stackengine}
\usepackage{framed}
\usepackage{scalerel}
\usepackage{array}
\usepackage{makecell}
\newcolumntype{x}[1]{>{\centering\arraybackslash}p{#1}}
\usepackage{tikz}
\usepackage{pgfplots}
\usetikzlibrary{shapes.geometric, shapes.misc, positioning, arrows, arrows.meta, decorations.pathreplacing, decorations.pathmorphing, patterns, angles, quotes, calc}
\usepackage{booktabs}
\usepackage{xfrac}
\usepackage{siunitx}
\usepackage{centernot}
\usepackage{comment}
\usepackage{chngcntr}
\usepackage[ruled]{algorithm2e}
\usepackage{algorithmic}
\usepackage{xparse}

\newtheorem{thm}{Theorem}

\newtheorem{lemma}[thm]{Lemma}
\newtheorem{cor}[thm]{Corollary}

\newtheorem{Def}[thm]{Definition}

\makeatletter
\def\thmhead@plain#1#2#3{%
  \thmname{#1}\thmnumber{\@ifnotempty{#1}{ }\@upn{#2}}%
  \thmnote{ {\the\thm@notefont#3}}}
\let\thmhead\thmhead@plain
\makeatother

\newtheorem{rem}[thm]{Remark}
\let\oldrem\rem
\renewcommand{\rem}{\oldrem\normalfont}
\newtheorem{ex}[thm]{Example}
\let\oldex\ex
\RenewDocumentCommand{\ex}{o}{%
  \IfNoValueTF{#1}
    {\oldex}
    {\oldex[#1]}%
  \normalfont
}

\newcommand{\bb}{\begin{equation}\begin{aligned}\hspace{0pt}}
\newcommand{\bbb}{\begin{equation*}\begin{aligned}}
\newcommand{\ee}{\end{aligned}\end{equation}}
\newcommand{\eee}{\end{aligned}\end{equation*}}
\newcommand*{\coloneqq}{\mathrel{\vcenter{\baselineskip0.5ex \lineskiplimit0pt \hbox{\scriptsize.}\hbox{\scriptsize.}}} =}

\newcommand\floor[1]{\left\lfloor#1\right\rfloor}
\newcommand\ceil[1]{\left\lceil#1\right\rceil}

\let\texteq\relax
\let\textleq\relax

\newcommand{\texteq}[1]{\stackrel{\mathclap{\scriptsize \mbox{#1}}}{=}}
\newcommand{\textleq}[1]{\stackrel{\mathclap{\scriptsize \mbox{#1}}}{\leq}}

\newcommand{\sumno}{\sum\nolimits}

\newcommand{\tcb}[1]{{\color{blue} #1}}

\newcommand{\R}{\mathds{R}}

\let\Tr\relax
\DeclareMathOperator{\Tr}{Tr}

\DeclareMathOperator{\co}{conv}

\DeclareMathOperator{\inter}{int}

\DeclareMathAlphabet{\pazocal}{OMS}{zplm}{m}{n}

\DeclareMathOperator{\ext}{ext}
\DeclareMathOperator{\pr}{Pr}

\DeclareMathOperator{\vol}{vol}

\newcommand{\lsmatrix}{\left(\begin{smallmatrix}}
\newcommand{\rsmatrix}{\end{smallmatrix}\right)}

\stackMath
\newcommand\xxrightarrow[2][]{\mathrel{%
  \setbox2=\hbox{\stackon{\scriptstyle#1}{\scriptstyle#2}}%
  \stackunder[5pt]{%
    \xrightarrow{\makebox[\dimexpr\wd2\relax]{$\scriptstyle#2$}}%
  }{%
   \scriptstyle#1\,%
  }%
}}

\newcommand{\tends}[2]{\xxrightarrow[\! #2 \!]{\mathrm{#1}}}

\stackMath

\makeatletter
\newcommand*\rel@kern[1]{\kern#1\dimexpr\macc@kerna}

\makeatletter
\def\reallywidetilde#1{\mathop{\vbox{\m@th\ialign{##\crcr\noalign{\kern3\p@}%
      \sortoftildefill\crcr\noalign{\kern3\p@\nointerlineskip}%
      $\hfil\displaystyle{#1}\hfil$\crcr}}}\limits}

\def\sortoftildefill{$\m@th \setbox\z@\hbox{$\braceld$}%
  \braceld\leaders\vrule \@height\ht\z@ \@depth\z@\hfill\braceru$}

\newcommand*\wtildea[1]{%
  \begingroup
    \settowidth{\dimen0}{$#1$}%
    \rlap{\resizebox{\dimen0}{\totalheight}{$\widetilde{\phantom{x\vphantom{#1}}}$}}%
  \endgroup
  #1}

\tikzset{meter/.append style={draw, inner sep=10, rectangle, font=\vphantom{A}, minimum width=30, line width=.8, path picture={\draw[black] ([shift={(.1,.3)}]path picture bounding box.south west) to[bend left=50] ([shift={(-.1,.3)}]path picture bounding box.south east);\draw[black,-latex] ([shift={(0,.1)}]path picture bounding box.south) -- ([shift={(.3,-.1)}]path picture bounding box.north);}}}
\tikzset{roundnode/.append style={circle, draw=black, fill=gray!20, thick, minimum size=10mm}}
\tikzset{squarenode/.style={rectangle, draw=black, fill=none, thick, minimum size=10mm}}

\definecolor{Blues5seq1}{RGB}{239,243,255}
\definecolor{Blues5seq2}{RGB}{189,215,231}
\definecolor{Blues5seq3}{RGB}{107,174,214}
\definecolor{Blues5seq4}{RGB}{49,130,189}
\definecolor{Blues5seq5}{RGB}{8,81,156}

\definecolor{Greens5seq1}{RGB}{237,248,233}
\definecolor{Greens5seq2}{RGB}{186,228,179}
\definecolor{Greens5seq3}{RGB}{116,196,118}
\definecolor{Greens5seq4}{RGB}{49,163,84}
\definecolor{Greens5seq5}{RGB}{0,109,44}

\definecolor{Reds5seq1}{RGB}{254,229,217}
\definecolor{Reds5seq2}{RGB}{252,174,145}
\definecolor{Reds5seq3}{RGB}{251,106,74}
\definecolor{Reds5seq4}{RGB}{222,45,38}
\definecolor{Reds5seq5}{RGB}{165,15,21}

\newcommand{\Depth}{3}
\newcommand{\Height}{3}
\newcommand{\Width}{3}

\makeatletter
\DeclareRobustCommand{\qed}{%
  \ifmmode 
  \else \leavevmode\unskip\penalty9999 \hbox{}\nobreak\hfill
  \fi
  \quad\hbox{\qedsymbol}}
\newcommand{\openbox}{\leavevmode
  \hbox to.77778em{%
  \hfil\vrule
  \vbox to.675em{\hrule width.6em\vfil\hrule}%
  \vrule\hfil}}
\newcommand{\qedsymbol}{\openbox}
\newenvironment{proof}[1][\proofname]{\par
  \normalfont
  \topsep6\p@\@plus6\p@ \trivlist
  \item[\hskip\labelsep\itshape
    #1.]\ignorespaces
}{%
  \qed\endtrivlist
}
\newcommand{\proofname}{Proof}
\makeatother

\usepackage{times,txfonts}

\let\tcb\relax

\begin{document}

\title{A Post-Quantum Associative Memory}



\author{Ludovico Lami}
\address{Institut f\"{u}r Theoretische Physik und IQST, Universit\"{a}t Ulm, Albert-Einstein-Allee 11, D-89069 Ulm, Germany}
\address{QuSoft, Korteweg-de Vries Institute for Mathematics, and Institute for Theoretical Physics, University of Amsterdam, Science Park, 1098 XG Amsterdam, the Netherlands}

\author{Daniel Goldwater}
\address{School of Mathematical Sciences,
University of Nottingham, University Park, Nottingham NG7 2RD, United Kingdom}

\author{Gerardo Adesso}
\address{School of Mathematical Sciences,
University of Nottingham, University Park, Nottingham NG7 2RD, United Kingdom}
\ead{gerardo.adesso@nottingham.ac.uk}

\begin{abstract}
Associative memories are devices storing information that can be fully retrieved given partial disclosure of it. We examine a toy model of associative memory and the ultimate limitations it is subjected to within the framework of general probabilistic theories (GPTs), which represent the most general class of physical theories satisfying some basic operational axioms. We ask ourselves how large the dimension of a GPT should be so that it can accommodate $2^m$ states with the property that any $N$ of them are perfectly distinguishable. Call $d(N,m)$ the minimal such dimension. Invoking an old result by Danzer and Gr\"unbaum, we prove that 
$d(2,m)=m+1$, to be compared with $O(2^m)$ when the GPT is required to be either classical or quantum. This yields an example of a task where GPTs outperform both classical and quantum theory exponentially. 
More generally, we resolve the case of fixed $N$ and asymptotically large $m$, proving that $d(N,m) \leq m^{1+o_N(1)}$ (as $m\to\infty$) for every $N\geq 2$, which yields again an exponential improvement over classical and quantum theories. Finally, we develop a numerical approach to the general problem of finding the largest $N$-wise mutually distinguishable set for a given GPT, which can be seen as an instance of the maximum clique problem on $N$-regular hypergraphs.
\end{abstract}
\maketitle

\section{Introduction}

A memory is a physical system which can be used to store some information which can later be retrieved. Memories can be complete, if all the information stored can be recovered at once; or incomplete, if only a part of it can be accessed. They can be perfect, if the retrieved pieces of information reproduce the original ones with probability one, or imperfect otherwise. The physical interest of designing an incomplete or imperfect memory is that in return for the loss of performance there might be an effective compression of the system size.

For example, the \emph{quantum random access encodings} of Ambainis et al.~\cite{Ambainis1998} (see also~\cite{VerSteeg2009}) allow for storing $2n$ classical bits into $n$ qubits, in such a way that \emph{any} given bit (but not all of them simultaneously) can be retrieved with probability $p\approx 0.79$. These memories are therefore incomplete and imperfect, but they allow for an effective compression of the physical system employed, as compared to the na\"ive encoding of $2n$ bits into $2n$ qubits, which is both complete and perfect.

The celebrated Hopfield network~\cite{Hopfield1982} is another example of an imperfect memory designed to model biological systems. An array of neurons is connected based on the desired information to be stored. The dynamics of the array result in attractors that precisely correspond to the stored states.
 The net effect is that upon being prepared in a certain initial configuration, the system often evolves towards the stored state that most resembles it. This mechanism amounts to an imperfect retrieval of the encoded information. The Hopfield network is in a certain sense an incomplete memory, because the recovery of a certain stored state can {take place} only if the initial configuration is sufficiently close to it. In other words; \emph{some} information about the stored state has to be disclosed if we want to retrieve the rest.

In this paper we want to study and characterise the ultimate physical limitations to the performance of incomplete memories. In order to achieve this, following recent developments~\cite{implausible, PVV, ultimate, XOR, cones-2, cones-3} we will utilise the formalism of \emph{general probabilistic theories} (GPTs). Within this mathematical framework, it is possible to model a vast family of physical theories, including classical probability theory, quantum mechanics; and more exotic theories such as generalised bits~\cite{Barrett-original}, spherical models~\cite{ultimate}, and Popescu--Rohrlich (PR) boxes~\cite{PR-boxes,Popescu1997}, to name a few~\cite{lamiatesi, Mueller2021, Plavala2021}.
We are particularly interested in finding out to what extent GPTs can exhibit an enhanced memory capacity compared to classical and quantum  theories.

The paper is organised as follows. In Section~\ref{sec:Problem} we expand upon and formalise the problem we are addressing --- establishing the relationship between a physical theory and the kind of incomplete memory which could be constructed within it. Section~\ref{sec:GPTs} reviews the GPT formalism and introduce some well known theories for reference. Section~\ref{sec:PairwisePerfect} holds the first main result: we  prove that a particular class of theories (those with hypercubic state spaces) are optimal for housing incomplete memories 
that can retrieve one lost bit. This is obtained by invoking a seminal result by Danzer and Gr\"unbaum~\cite{Danzer1962}. In Section~\ref{sec:beyondPairwise} we begin to search for the optimal theory when 
the number of bits to be retrieved is arbitrary. There we prove our second main result, 
which gives the scaling of the minimal dimension of a GPT that can host very large incomplete memories capable of retrieving a fixed number of lost bits. In both of our main results, GPTs are shown to outperform classical and quantum associative memories \emph{exponentially}. In Section~\ref{sec:beyondNum} we recast the task of determining  the largest $N$-wise mutually distinguishable set for a given GPT as the convex problem of finding the maximum $N$-clique on an $N$-regular hypergraph.
Finally, we conclude in Section~\ref{sec:Discussion}.

\subsection{The problem} \label{sec:Problem}

We will focus on the simplest type of incomplete perfect memory, whose general working principle is as follows. We begin by storing in it an $m$-bit string $x\in \{0,1\}^m$ by means of a suitable encoding. The value of $x$ is then forgotten, with the only remaining record being stored in the memory. Later, we are given $N$ $m$-bit strings {$Y = \{y_1,\dots,y_N\}$}, with $y_i\in \{0,1\}^m,\,\, i=1,\dots,N$, with the promise that one of these matches the original string, $y_i=x$. Our task is to determine $x$ by making a suitable measurement on our device. The memory is called incomplete if the largest achievable $N$ satisfies $N_{\max}<2^m$, and perfect if the recovery can be achieved with unit success probability for all choices of $x$ and $Y$, with the constraint that $Y$ has cardinality $N$.

If we model physical systems in terms of GPTs, the problem can be equivalently seen as asking for a GPT $A$ and an encoding function $\rho:\{0,1\}^m \to \Omega_A$, with $\Omega_A$ being the state space on $A$, such that any $N$ distinct states $\rho(x_1),\ldots, \rho(x_N) \in \Omega_A$ are perfectly distinguishable. Equivalently, we could demand the existence of $2^m$ states $\rho_1,\ldots,\rho_{2^m}\in \Omega_A$ that are $N$-wise mutually distinguishable, meaning that any $N$ of them are (jointly) perfectly distinguishable. We will formalise our notions of {\emph{perfect distinguishability} and \emph{mutual $N$-wise distinguishability}} in Definitions~\ref{def:perfectly_distinguishable} and~\ref{def:mutually_distinguishable}. Especially the former concept has attracted considerable interest recently~\cite{Arai2019, Yoshida2020}.

In order to assess the capacity of a memory in system $A$, we will want to quantify the effective \emph{compression} operated by the encoding $\rho$. This presents a problem; we cannot count the number of bits or qubits in the system $A$, because this will be modelled by a GPT that is in general neither classical nor quantum. However, there \emph{is} a universal way to quantify how `large' a GPT is: its dimension. Since a classical $m$-bit system can be represented by a GPT of dimension $2^m$, we could employ the logarithm $\log_2 d$ of the dimension $d=\dim V$ of a certain GPT as an effective measure of its memory capacity. Along the same lines, we could employ the \emph{compression factor}
\begin{equation}
\kappa = \frac{\text{\# encoded bits}}{\log_2 \dim V} = \frac{m}{\log_2 d}
\label{eq:compression}
\end{equation}
to assess the quality of the scheme. {The {$N$-wise} compression factor, denoted with $\kappa(N, m)$, is the maximum such $\kappa$ that is achievable with all possible GPTs. It is a universal function of the pair $(N,m)$, as the optimisation does away with the degree of freedom represented by the choice of the underlying GPT. Clearly, it is given by $\kappa(N,m)=\frac{m}{\log_2 d(N,m)}$, where $d(N,m)$ is the minimum $d$ such that a $d$-dimensional GPT hosting an $N$-wise mutually distinguishable set of states of cardinality $2^m$ can be found.}
This discussion allows us to precisely state our problem as follows.

\begin{center}
\begin{minipage}{.8\textwidth}
{\textbf{Problem.}
For all pairs of positive integers $N,m$, compute $\kappa(N,m)$, i.e.\ determine the minimum dimension of a GPT that can host an $N$-wise mutually distinguishable set of states of cardinality $2^m$.}
\end{minipage}
\end{center}

\begin{rem} \label{rem:NwiseScaling}
Given $n$ strings encoded into a system $A$, and granted that those states are all pairwise perfectly distinguishable, it will take (at most) 
$n-1$ measurements to uniquely identify the desired state via a tournament-like method, granted that the measurements are non-disturbing (see~\cite{Chiribella-pur}). Alternatively, if the measurements \emph{are} disturbing, we would require an equivalent number of copies of the system $A$. If instead those states are mutually $N$-wise distinguishable, these numbers reduce to {$\ceil{\frac{n-1}{N-1}}$}.

\end{rem}


Before we get to the formalism of GPTs, {by means of which we will explore more} exotic theories, we can first examine the performance of the most familiar {ones}: quantum and classical mechanics. For these examples we choose $N=2$, so that we are finding the maximum number of pairwise distinguishable states which a system can store.

\begin{itemize}
\item \emph{Classical theory.} If two classical probability distributions over an alphabet $\pazocal{X}$ are pairwise perfectly distinguishable it means that they have disjoint supports inside $\pazocal{X}$. If $2^n$ probability distributions on $\pazocal{X}$ are pairwise perfectly distinguishable, we deduce that their supports $Y_i$ are all disjoint, and therefore $d= |\pazocal{X}| \geq \sum_{i=1}^{2^n} |Y_i| \geq 2^n$. Expressed in words, this entails that in order to accommodate $2^n$ pairwise perfectly distinguishable states, a classical system must have dimension at least $2^n$. This lower bound is trivially tight, so that the $N=2$ compression factor of classical theories is precisely $1$.

\item \emph{Quantum theory.} If $2^n$ quantum states are pairwise perfectly distinguishable, their supports must be pairwise orthogonal. This means that the total dimension \emph{of the Hilbert space} is at least $2^n$. Since the dimension of quantum mechanics as a GPT is the \emph{square} of the Hilbert space dimension (cf.~\eqref{dim quantum}), we see that a quantum system capable of accommodating $2^n$ pairwise perfectly distinguishable states must have dimension at least $2^{2n}$. Again, this lower bound is easily seen to be tight{, entailing that the $N=2$ compression factor for quantum theory is precisely $1/2$.}
\end{itemize}

{Since their compression factors are at most $1$}, classical \emph{as well as quantum} theory perform rather poorly at the task we are interested in here.

In~\cite{PR-boxes} Popescu and Rohrlich famously showed that a hypothetical `super-quantum' theory could outperform quantum mechanics at non-local tasks. However, results in~\cite{Lee2015a, Lee2016a} indicate that such exotic theories may not beat quantum theory in terms of computational capacity. Here we will see how other theories fare at the task of {implementing an} associative memory and, in particular, seek out the optimal theory --- that  with the highest compression ratio defined above.

\section{General probabilistic theories}\label{sec:GPTs}

{
Throughout this Section we will formally introduce and discuss general probabilistic theories. We point the interested reader to Ref.~\cite{lamiatesi, mueller2020, Plavala2021} for more details and a thorough operational justification of the construction described here.

We start by fixing some terminology. A subset $C\subseteq V$ of a finite-dimensional, real vector space $V$ is called a \textbf{cone} if it is closed under positive scalar multiplication. It is called a \textbf{proper cone} if in addition it is (i)~convex; (ii)~salient, that is, $C\cap (-C)=\{0\}$; (iii)~spanning, meaning that $C-C=V$; and (iv)~topologically closed.\footnote{Since we are in finite dimension, there is a unique Hausdorff topology on $V$, which we do not need to specify. For instance, it is induced by any Euclidean norm.}

In what follows, we will denote the dual vector space to $V$, i.e.\ the space of linear functionals $V\to \R$, with $V^*$. If $C\subset V$ is a cone, we can construct its \textbf{dual cone} inside $V^*$ as $C^*\coloneqq \left\{f\in V^*:\, f(x)\geq 0\ \forall\, x\in C \right\}$. If $C$ is proper then so is $C^*$, and moreover $C^{**} = C$ modulo the canonical identification $V^{**}=V$. A functional $f\in C^*$ is also said to be positive; it is \textbf{strictly positive} if $f(x)>0$ for all $x\in C$ with $x\neq 0$. It can be verified that strictly positive functionals are precisely those in the topological interior of $C^*$, denoted by $\inter\left(
C^*\right)$.


\begin{Def}[General probabilistic theories] \label{GPT_def}
A \textbf{general probabilistic theory} (GPT) is a triple $(V, C, u)$ consisting of a real, finite-dimensional vector space $V$, a proper cone $C\subset V$, and a strictly positive functional $u\in \inter\left(C^*\right)$, called the \textbf{order unit}. We call $d\coloneqq \dim V$ the \textbf{dimension} of the GPT, and $\Omega\coloneqq C\cap u^{-1}(1) = \left\{ x\in C:\, u(x)=1\right\}$ its \textbf{state space}. A \textbf{pure state} is an extreme point\footnote{An extreme point of a convex set $X$ is a point $x\in X$ such that $x=py+(1-p)z$ for $p\in (0,1)$ and $y,z\in X$ implies that $y=z=x$. The set of extreme points of $X$ will be denoted by $\ext(X)$.} of $\Omega$. An \textbf{effect} is a functional $e\in V^*$ such that $e(\omega)\in [0,1]$ for all $\omega\in \Omega$. We will denote the set of effects with $E = C^*\cap \left(u-C^*\right)$. A \textbf{measurement} is a finite collection $(e_i)_{i\in I}$ of effects $e_i\in E$ such that $\sum_{i\in I} e_i = u$.
\end{Def}

\begin{rem}
The restriction to finite-dimensional spaces is made for purely technical reasons, as it simplifies the treatment considerably. However, the GPT framework makes perfect sense in infinite dimension as well --- in fact, GPTs were initially conceived to accommodate also this case~\cite{LUDWIG,FOUNDATIONS,Davies-1970} (see also~\cite[Chapter~1]{lamiatesi}).
\end{rem}

The state space $\Omega$ as well as the set of effects $E$ of a given GPT are always compact convex sets. As such, they can be equivalently described as the convex hulls of their extreme points (in the case of $\Omega$, these are just the pure states of the theory).
Two extreme points of $E$ are always $0$ and the order unit $u$.

\medskip
\noindent {\bf Note.}
It is worthwhile to point out some subtleties concerning the interpretation of the above definition of a GPT that should be kept in mind:
\begin{itemize}
    \item We implicitly assume the \emph{no restriction hypothesis}~\cite{no-restriction}. This states that all abstract measurements as constructed in Definition~\ref{GPT_def} are actually physically implementable, and entails that defining the state space $\Omega$ of a theory is sufficient to completely determine its local structure. We deem it a fairly natural assumption, since GPTs are operationally motivated in the first place -- state and effect spaces can be thought of as mutually defining -- and the class of restricted GPTs can do no better than the class of unrestricted GPTs for this particular task.
    \item We are considering only those theories with finite-dimensional state spaces (for an exploration beyond this, see~\cite[Chapter~1]{lamiatesi}).
    \item We are only dealing with the \emph{reliable} states and effects for a theory. Operationally, this is equivalent to having preparation and measurement procedures which always behave as desired (for example, we can produce specific states deterministically).
    \item We are not examining non-local correlations or entanglement-like features available in different GPTs, which are often the subject of enquiry in the GPT literature~\cite{telep-in-GPT, Shahandeh2021, Schmid2020, DAriano2020, cones-1, cones-2, cones-3}. However, although we are only considering the geometries of single systems, it is worth emphasising that these do impact upon which non-local correlations can be attained~\cite{Janotta2011, Short2010, cones-1, cones-2, cones-3}.
\end{itemize}

}

\subsection{Some Example Theories}

\begin{ex}[Classical probability theory] \label{ex class}
States in a classical probability theory are simply probability distributions over some finite alphabet $\pazocal{X}$. The corresponding GPT will have dimension $d=|\pazocal{X}|$, where $|\pazocal{X}|$ is the size of $\pazocal{X}$. Formally, it can be defined as a triple $\big(\R^{d},\, \R^{d}_{+},\, u\big)$, where $\R^{d}_{+}\coloneqq\{ x\in \R^{d}:\, x_{i}\geq 0\ \forall\, i=1,\ldots, d\}$ is just the positive orthant, and the unit effect is a functional acting as $u(y)=\sum_{i=1}^{d} y_i$ for all $y\in \R^{d}$. The state space is therefore formed by all non-negative vectors $x\in \R_+^d$ such that $u(x) =\sum_i x_i=1$; geometrically, this set is shaped as a simplex with $d$ vertices, which we denote by $\mathcal{S}_d$.
\end{ex}

\begin{ex}[Quantum mechanics] \label{ex QM}
The quantum mechanical theory of a $k$-level system can also be phrased in the GPT language. Formally, we can define it as the triple $\left( \mathrm{H}_k,\, \mathrm{PSD}_k,\, \Tr \right)$, where $\mathrm{H}_k$ is the real vector space of $k\times k$ Hermitian matrices, $\mathrm{PSD}_k$ is the cone of $k\times k$ positive semidefinite matrices{, and $\Tr$ is the trace functional.} Observe that the \emph{real} dimension of $k$-level quantum mechanics is
\bb
\dim \mathrm{H}_k = k^2.
\label{dim quantum}
\ee
\end{ex}

\begin{ex}[$n$-gon theories]\label{ex:ngon_theories}
$n$-gon theories (sometimes referred to as polygon theories), are those in which the state space is described by a regular $n$-sided polygon. These theories are well studied~\cite{Massar2014a, Janotta2011, Heinosaari2019, Kobayshi2017, Pfister2013}, and contain the local structure of Popescu--Rohrlich boxes as a particular case ($n=4$). Interestingly, there is a general difference between those in which $n$ is odd and those in which $n$ is even: for odd $n$, the theories are strongly self-dual, meaning that the {dual cone $C^*$ is isomorphic to $C$ via an isomorphism mediated by a positive definite scalar product. For even $n$, the theories are only weakly self-dual, meaning that $C$ and $C^*$ are merely linearly isomorphic.}
\end{ex}

\begin{rem}
One particularly nice property of $n$-gon theories is that they give a (restricted) version of both quantum and classical theories in limiting cases. In the limiting case of $n=2$, the polygon collapses to the line segment; this can be taken to represent a stochastic classical bit (such as a coin). In the other extreme, at $n=\infty$, the `polygon' describes a circle --- which can be thought of representing a slice through the Bloch sphere, such as the slice of states with real-valued coefficients $|\psi\rangle = \alpha|0\rangle+\beta|1\rangle$, with $\alpha, \beta \in \R$ and $\alpha^2+\beta^2=1$.
\end{rem}


{
\subsection{Perfect distinguishability}

Now that we have a rigorous definition of GPT in place, we can also give a precise meaning to the various notions of perfect distinguishability employed in this paper. We start with the basic definition of perfect distinguishability for a set of states in a GPT. For additional details and further motivation we refer the reader to~\cite{Arai2019, Yoshida2020}.

\begin{Def}[Perfect distinguishability] \label{def:perfectly_distinguishable}
Let $(V,C,u)$ be a GPT with state space $\Omega$. We say that some finitely many states $\{\omega_i\}_{i\in I}\subseteq \Omega$ are \textbf{perfectly distinguishable} if there exists a measurement $(e_i)_{i\in I}$ such that $e_i(\omega_j) = \delta_{i,j}$ for $i,j\in I$.
\end{Def}

We can now give a notion of mutual distinguishability for sets of states.

\begin{Def}[Mutual $\boldsymbol{N}$-wise distinguishability] \label{def:mutually_distinguishable}
Let $(V,C,u)$ be a GPT with state space $\Omega$. A set of states $\mathcal{S}\subseteq \Omega$ is said to be \textbf{mutually $\boldsymbol{N}$-wise distinguishable} if every subset $S\in\mathcal{S}$ of cardinality $|S|=N$ is
perfectly distinguishable as per Definition~\ref{def:perfectly_distinguishable}. If $N=2$ we also say that the states in $\mathcal{S}$ are \textbf{pairwise perfectly distinguishable}.
\end{Def}

\begin{rem}
The fact that $\{\omega_1,\omega_2\}$ and $\{\omega_2,\omega_3\}$ are separately perfectly distinguishable does \emph{not} imply, in general, that $\{\omega_1,\omega_2,\omega_3\}$ are  perfectly distinguishable. More generally; the union of some sets which are $N$-wise distinguishable is not necessarily \emph{mutually} $N$-wise distinguishable itself.
\end{rem}


In what follows we will be interested in the minimal GPT dimension that is needed in order to achieve mutually $N$-wise distinguishable sets with a prescribed number of elements, or, vice versa, in the maximal number of elements that a mutually $N$-wise distinguishable set of states can have in GPTs of a fixed dimension. We thus formalise the following definition.

\begin{Def} \label{def:compression-factor}
For two positive integers $N,m$, we denote with $d(N,m)$ the \tcb{minimum dimension $\dim V_A$ among all GPTs $A=(V_A,C_A,u_A)$ having the property that} 
the corresponding state space $\Omega_A=C_A\cap u_A^{-1}(1)$ contains a set of mutually $N$-wise distinguishable states of cardinality $2^m$. The corresponding \textbf{compression factor} is defined by
\begin{equation}
    \kappa(N,m) \coloneqq \frac{m}{\log_2 d(N,m)}\, .
\end{equation}
\end{Def}

If we accept the assumptions leading to the GPT framework as we have defined it above, calculating or estimating $\kappa(N,m)$ from above (equivalently, calculating or estimating $d(N,m)$ from below) amounts to establishing the ultimate physical bounds to the compression of information realised by an incomplete but perfect memory. The rest of the paper is devoted to the understanding of these quantities and to their exact computation in a few interesting cases.

We start by looking at the most extreme case, that where the memory is in fact complete, i.e.\ the information can be retrieved. This corresponds to setting $N=2^m$. In this case,  even GPTs do not grant any advantage over classical probability theory.

\begin{lemma} \label{lemma:kappa=1}
For all positive integers $m$, it holds that $d(2^m,m)=2^m$ and hence $\kappa(2^m,m)=1$. In other words, there exists a GPT (namely, classical probability theory) of dimension $2^m$ hosting $2^m$ perfectly distinguishable states, but no GPT of smaller dimension enjoying that same property.
\end{lemma}

\begin{proof}
Since perfectly distinguishable states must be linearly independent,  the dimension of the host vector space of any GPT accommodating $2^m$ perfectly distinguishable states must be at least $2^m$.
\end{proof}

The above Lemma~\ref{lemma:kappa=1} is slightly disappointing, as it tells us that even GPTs cannot perform better than classical probability theory at the implementation of a perfect and complete memory. However, this state of affairs changes dramatically when we consider smaller values of $N$, i.e.\ when we look instead at perfect but incomplete memories. We will see how this is possible in the next Section.
}

\section{Pairwise distinguishability}\label{sec:PairwisePerfect}


In this Section we show that a compression factor much larger than $1$, and indeed of order $m$ up to logarithmic factors, is achievable when $N=2$. Even more, we give an exact expression for the function $\kappa(2,m)$.

\begin{thm} \label{thm:pairwise}
For all positive integers $m$, it holds that $d(2,m)=m+1$ and hence
\begin{equation}
\kappa(2,m) = \frac{m}{\log_2 (m+1)}\, .
\label{kappa_2_m}
\end{equation}
In other words, there exists a GPT of dimension $m+1$ hosting \tcb{$2^m$} pairwise distinguishable states, but no GPT of dimension $m$ or lower enjoying this same property.
\end{thm}

The above result\tcb{, whose proof can be found at the end of Section~\ref{subsec:hypercube},} is remarkable because it provides an example of a task at which GPTs outperform both classical and quantum theories dramatically. In fact, as we saw in Section~\ref{sec:Problem} the compression factor $\kappa(2,m)$ for such theories is just a constant, while Theorem~\ref{thm:pairwise} tells us that in the GPT world it can be made much larger, of the order of $m$ (up to a logarithmic factor). Another notable aspect of Theorem~\ref{thm:pairwise} is that it does not report an estimate but rather an exact computation of the figure of merit that is of interest here, thus establishing the ultimate physical limits to this very simple type of incomplete (perfect) memory.

The discussion and proof of Theorem~\ref{thm:pairwise} occupies the rest of the present Section. More in detail, in Section~\ref{subsec:d=3} we discuss the simplest non-trivial case of $3$-dimensional GPTs, proving with a delightfully simple argument that $d(2,2)=3$, or equivalently $\kappa(2,2) = 2/\log_2(3)$. Section~\ref{subsec:hypercube} is devoted to the presentation of the general construction that achieves the best compression factor~\eqref{kappa_2_m} among all GPTs. In~\ref{appendix:DG} we revisit the proof of the Danzer--Gr\"unbaum theorem, showing that it implies directly the optimality of the above construction.

\subsection{Limits in $d=3$} \label{subsec:d=3}

Before commencing, 
\tcb{a note} on geometric terminology. We say that a hyperplane $V\subset \R^n$ \emph{supports} a set $X$ in a point $x\in V\cap X$ if $V$ touches $X$ in $x$ without `cutting through' it, in other words, if the whole $X$ lies in one of the two closed half-spaces determined by $V$, with $x\in V\cap X$. Formally:

\begin{Def}
Let $X\subseteq \R^n$ be a subset of a Euclidean space. We say that a hyperplane $V\subset \R^n$ supports $X$ in a point $x\in X$ if: (i)~$x\in V$; and (ii)~$X$ is entirely contained inside one of the closed half-spaces determined by $V$.
\end{Def}

Let us consider a $d$-dimensional GPT with state space $\Omega\subset \R^{d-1}$ and the set of (distinct) states $\{\rho_i\}_{i=1,\ldots, 2^m}\subset \Omega$. Assume that any pair $\{\rho_i, \rho_j\}$ with $i\neq j$ is perfectly distinguishable by means of a measurement $(e_{ij},\, u-e_{ij})$, as per Definition~\ref{def:perfectly_distinguishable}. Explicitly, this means that $e_{ij}(\rho_i)=1$ and $e_{ij}(\rho_j)=0$. Note that the set of vectors $v$ such that $e_{ij}(v)=0$ and the set of vectors $w$ such that $e_{ij}(w)=1$ form two parallel hyperplanes $V$ and $W$. Note that $\rho_j\in V$ and $\rho_i\in W$. Clearly, since $0\leq e_{ij}(\omega)\leq 1$ for all states $\omega$, the whole $\Omega$ lies between $V$ and $W$. We can say that $W$ and $V$ support the state space $\Omega$ in $\rho_i$ and $\rho_j$, respectively. Vice versa, this condition is entirely equivalent to $\rho_i$ and $\rho_j$ being perfectly distinguishable. To get a clear geometric intuition it is instructive to explore  the special case where the state space is $2$-dimensional; with our convention, this corresponds to the case where $d=3$, because the global GPT will feature a $3$-dimensional cone whose section is our $2$-dimensional state space.

We thus consider a $3$-dimensional GPT with states confined to a set $\Omega\subset \R^{2}$. {The situation is as depicted in Figure~\ref{fig:support}. The two states $\rho_i,\rho_j\in \Omega$ in Figure~\ref{fig:support} are indeed perfectly distinguishable, because the entire set $\Omega$ is enclosed between two parallel lines supporting it in $\rho_i$ and $\rho_j$, respectively. However, one can see that not all pairs among the $6$ states marked with black dots can be perfectly distinguishable.\footnote{Indeed, for example the two dots at the bottom of the grey figure are not. This makes sense, because we see from Theorem~\ref{thm:pairwise} that in dimension $3$ there can be at most $2^{3-1}=4$ states with such property.}}

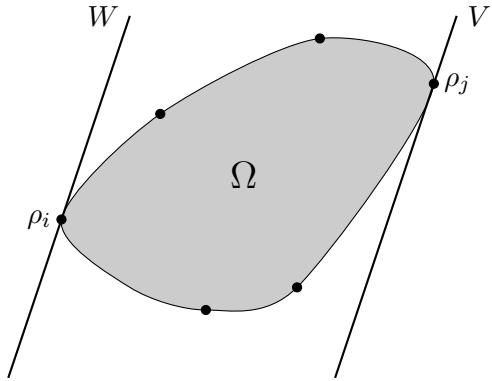
\begin{figure}[ht] \centering
\begin{tikzpicture}
\filldraw [fill=gray!40, even odd rule] plot [smooth cycle] coordinates {(0,0) (1.2,.3) (3,3) (1.5,3.6) (-.6,2.6) (-1.9,1.2) (-1,.3)};
\node at (.5,1.8) {\Large $\Omega$};
\draw[fill=black] (0,0) circle (.6mm);
\draw[fill=black] (1.2,.3) circle (.6mm);
\draw[fill=black] (3,3) node[anchor=west] {$\rho_j$} circle (.6mm);
\draw[fill=black] (1.5,3.6) circle (.6mm);
\draw[fill=black] (-.6,2.6) circle (.6mm);
\draw[fill=black] (-1.9,1.2) node[anchor=east] {$\rho_i$} circle (.6mm);
\draw[thick] (1.7,-.9) -- ++(1.6,4.8) node[anchor=west] {$V$};
\draw[thick] (-2.6,-.9) -- ++(1.6,4.8) node[anchor=east] {$W$};
\end{tikzpicture}
\caption{A set of states in a two-dimensional GPT with state space $\Omega$ --- here we are viewing the state space top-down. If $\rho_i$ and $\rho_j$ are perfectly distinguishable then the two lines $W$ and $V$ \emph{support} the state space in $\rho_i$ and $\rho_j$, respectively.}
\label{fig:support}
\end{figure}

Let us make this discussion a bit more rigorous. Assume that we are given $k$ states $\rho_1,\ldots, \rho_k\in \Omega$, with the promise that they are pairwise perfectly distinguishable. We can ask ourselves: \emph{how large can $k$ be?} The convex hull of $\rho_1,\ldots, \rho_k$ will naturally form a polygon $P\subseteq \Omega$. In fact, we have that every $\rho_i$ must correspond to a vertex of $P$ in order for the perfect distinguishability condition to be obeyed.
Consider now two neighbouring vertices $\rho_i,\rho_{i+1}$ of $P$, as well as the edge connecting them. Call $\alpha_i,\alpha_{i+1}$ the internal angles of $P$ at vertices $\rho_i,\rho_{i+1}$. It can be shown that, in order for $\rho_i, \rho_{i+1}$ to be perfectly distinguishable, it has to hold that $\alpha_i+\alpha_{i+1}\leq \pi$ (cf.~Figure~\ref{two_dim_fig}). Summing over $i=1,\ldots, k$, with the convention that $k+1\equiv 1$, we obtain that
\bb
k\pi \geq \sum_{i=1,\ldots, k} (\alpha_i+\alpha_{i+1}) = 2 \sum_i \alpha_i\, .
\ee
The sum on the right-hand side is just the sum of all internal angles of a convex polygon with $k$ vertices. From elementary geometry, this is well known to be $(k-2)\pi$. Therefore, we obtain the inequality
\bb
k\pi \geq 2(k-2)\pi\, ,
\ee
which yields immediately $k\leq 4=2^2$, in line with Theorem~\ref{thm:pairwise}. This bound is tight, because the four vertices of a square state space correspond to pairwise perfectly distinguishable states --- a more general version of this latter statement will be proved in the next Section.

\begin{figure}[ht] \centering
\begin{tikzpicture}
\coordinate (a) at (0,0);
\coordinate (b) at (3,0);
\coordinate (a') at (1,1.5);
\coordinate (b') at (3.4,1.5);
\draw[fill=black] (a) node[left] {$\rho_i$} circle (.6mm);
\draw[fill=black] (b) node[right] {$\rho_{i+1}$} circle (.6mm);
\draw[thick] (a) -- (b);
\draw[thick, dotted] (-.5,-1) -- ++(2,4);
\draw[thick, dotted] (2.5,-1) -- ++(2,4);
\draw[thick] (a) -- (a');
\draw[thick, dashed] (a') -- ++(.5,.75);
\draw[thick] (b) -- (b');
\draw[thick, dashed] (b') -- ++(.2,.75);
\pic["$\alpha_i$", draw, <->, angle eccentricity=1.4, angle radius=.7cm] {angle=b--a--a'};
\pic["$\alpha_{i+1}$", draw, <->, angle eccentricity=1.5, angle radius=.6cm] {angle=b'--b--a};
\node at (2.3,1.8) {\Large $P$};
\end{tikzpicture}
\caption{A geometric sketch of a possible proof for the simplest non-trivial case $N=2$, $d=3$.}
\label{two_dim_fig}
\end{figure}
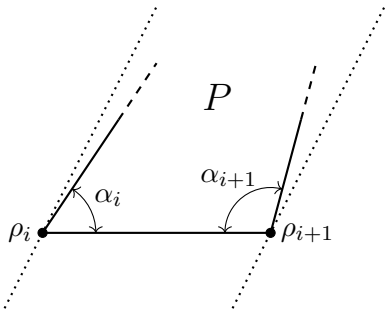

\subsection{Generalisation to arbitrary dimension and optimality} \label{subsec:hypercube}

We now set out to generalise the analysis to any dimension.\footnote{In a related spirit, theories using hyperspheres of generalised dimension, so-called $D$-balls, are discussed in~\cite{Masanes2014, Krumm2019}; the authors aim to isolate the 3-sphere as a the necessary state space for quantum theory based on physical requirements. See also~\cite{ultimate} for a different use of spherical theories.}
In light of the geometric construction discussed in Section~\ref{subsec:d=3} (the situation is entirely analogous to that depicted in Figure~\ref{fig:support} for $d=3$), we can reformulate our problem as follows:
\begin{center}
\begin{minipage}{.8\textwidth}
\textbf{Problem (reformulation).} Determine the minimum $d_m=d(2,m)$ such that there exists a set $X=\{\rho_1,\ldots,\rho_{2^m}\}\subset \R^{d_m-1}$ with the following property: for any two distinct $\rho_i,\rho_j\in X$, there are two parallel hyperplanes $W,V\subset \R^{d_m-1}$ of which one supports $X$ in $\rho_i$ and the other supports $X$ in $\rho_j$.
\end{minipage}
\end{center}

We now explain how to achieve a construction with the above properties in dimension $d =m+1$. The argument is quite simple, and it is worthwhile explaining it in words before delving into the mathematical formalism. The state space of the GPT we pick to achieve the bound is shaped as a hypercube of dimension $m$. Since the whole theory includes also multiples of normalised spaces, its dimension is in fact $m+1$. The $2^m$ states we choose correspond to the vertices of the hypercube. The crucial point now is that any two distinct vertices will be sitting each on one of two parallel hyperplanes that enclose the whole state space. Those hyperplanes, that are spanned by two opposite faces of the hypercube, will define the binary measurement needed to discriminate the states in question. This bit of reasoning already shows that any two vertices of the hypercube indeed represent perfectly distinguishable states.

We now make this argument rigorous. Construct the GPT $\left(\R^{m+1}, C_{g,m}, u\right)$, where
\bb
C_{g,n} \coloneqq \left\{ (x_0,x_1,\ldots, x_n)^\intercal\in \R^{m+1}:\ x_0 \geq \max_{1\leq i\leq n} |x_i| \right\}
\label{cubic_cone}
\ee
and moreover $u\left( (x_0,x_1,\ldots, x_m)^\intercal \right) \coloneqq x_0$. The state space of this GPT is clearly a hypercube of dimension $m$. Now, for $\epsilon\in \{\pm 1\}^n$, define $\rho_\epsilon\in \R^{m+1}$ by
\bb
(\rho_\epsilon)_i \coloneqq \left\{\begin{array}{ll} 1 & \text{if $i=0$,} \\[1ex] \epsilon_i & \text{if $i\geq 1$.} \end{array}\right.
\ee
Note that there are exactly $2^m$ distinct choices for $\epsilon$. We deduce that the $2^m$ states $\rho_\epsilon$ are pairwise perfectly distinguishable. To see why, consider $\epsilon,\epsilon'\in \{\pm 1\}^n$ that are distinct. Then, they will differ at some position $i\in \{1,\ldots, m\}$. Without loss of generality, we can assume that $\epsilon_i=+1$ and $\epsilon'_i = -1$. Now, consider the two-element collection $(e_i,u-e_i)$, where the functional $e_i$ is defined by
\bb
e_i \left( (x_0,x_1,\ldots, x_n)^\intercal \right) \coloneqq \frac{x_0+x_i}{2}\, .
\ee
Note that for all $x\in C_{g,m}$ we have that $0\leq e_i(x)\leq u(x)$; hence, the collection $(e_i,u-e_i)$ defines a binary measurement. It is now elementary to verify that
\bb
e_i(\rho_\epsilon) = 1\, , \quad (u-e_i)(\rho_\epsilon) = 0\, , \quad e_i(\rho_{\epsilon'}) = 0\, , \quad (u-e_i)(\rho_{\epsilon'}) = 1\, .
\ee
These are precisely the conditions needed to ensure that $\rho_\epsilon$ and $\rho_{\epsilon'}$ are perfectly distinguishable.\footnote{
These `hypercubic' theories have been employed in a similar spirit by Ver Steer and Wehner~\cite[Claim~6.2]{VerSteeg2009} to construct superior random access codes.}

We have therefore constructed a GPT of dimension $m+1$ which is capable of accommodating $2^m$ pairwise perfectly distinguishable states; {hence,
\begin{equation}
    d(2,m) \leq m+1\, ,\qquad \text{or}\qquad \kappa(2,m)\geq \frac{m}{\log_2(m+1)}\, .
\end{equation}}
What is remarkable here is that those limits far exceed the capabilities of both classical and quantum theory --- each of these have an exponential scaling in the number of dimensions required to store $m$ bits, whereas hypercubic theories scale only linearly. {Equivalently, the compression factor for the case $N=2$ of pairwise perfect distinguishability is at most $1$ for classical and quantum theory, but scales almost linearly in $m$ (up to logarithmic factors) for the best conceivable GPT. This demonstrates a sort of \emph{exponential advantage} of general GPTs over classical and quantum theories.}

It remains to show that the above construction is optimal. From the mathematical standpoint, this is highly non-trivial. To overcome this hurdle, we exploit the reformulation of the problem presented in Section~\ref{subsec:hypercube}: in that form, the problem was posed for the first time by Klee~\cite{KLEE} and was solved not long after by Danzer and Gr\"unbaum~\cite{Danzer1962}. Their solution shows that $d_m=m+1$ is a minimum for any $m$. We restate their result for our convenience below.\footnote{Our poor knowledge of German meant that we employed a translation of the original paper, realised by Rolf Schneider.}

\begin{thm}[{Danzer--Gr\"unbaum~\cite{Danzer1962}}] \label{thm:DG}
For a positive integer $n$, the maximum cardinality of a set $X\subset \R^n$ such that for any two distinct $x_1,x_2\in X$ there are two parallel hyperplanes $V_1,V_2\subset \R^n$ with the property that $V_i$ supports $X$ in $x_i$ ($i=1,2$) is precisely $2^n$. This cardinality is achieved by the set of vertices of a hypercube. \tcb{Moreover, up to affine operations the set of vertices of a hypercube is the \emph{only} set of points with this property having maximal cardinality.}
\end{thm}

For the interested reader, in~\ref{appendix:DG} we present a brief but self-contained account of, and homage to, the beautiful proof by Danzer and Gr\"unbaum~\cite{Danzer1962}; see also~\cite[Chapter~17]{THEBOOK}. \tcb{We can now formally deduce the proof of Theorem~\ref{thm:pairwise} as a simple corollary of the above result.

\begin{proof}
By Theorem~\ref{thm:DG}, the dimension $d'$ of any space capable of hosting $2^m$ points with the property discussed in the problem reformulation 
on p.~10 satisfies that $d'\geq m$. The dimension of the corresponding GPT is obtained by adding one, so that $d(2,m)\geq m+1$. The above example, also on p.~10, achieves this bound. Hence $d(2,m)=m+1$, completing the proof.
\end{proof}
}



\section{Perfect distinguishability beyond pairwise: Asymptotic results}\label{sec:beyondPairwise}

In the previous Section we have established the maximum number of pairwise perfectly distinguishable states which can be housed in a GPT, and hence the limits to the capacity of an associative memory of the type described in our introduction, in the case where $N=2$. The situation is much less clear for $N\geq 3$, in which case we cannot exhibit an explicit expression for $\kappa(N,m)$, nor a tight general estimate. However, in Theorem~\ref{thm:beyond_pairwise} below we determine the exact asymptotics in $m$ for every fixed $N$. Before we do so, it is instructive to see how a na\"{i}ve generalisation of the hypercube construction actually fails to yield an exact computation of $\kappa(N,m)$.




\subsection{A na\"{i}ve generalisation and its fall}

At first, we could hope that a simple generalisation of the hypercube construction may work. To explain how to obtain such a generalisation, we start by observing that from the geometric standpoint a hypercube can equivalently be seen as a Cartesian product of segments.
Indeed, the extreme points of a simple line segment can be thought of as having coordinates $(\pm 1)$. A square, of having all four combinations of $(\pm 1, \pm 1)$, and so on for cubes, and hypercubes in any dimension. This operation of combining vertices by concatenating their coordinates corresponds precisely to the geometric construction of the Cartesian product. 
Such a construction can be translated into the world of GPTs in a fully general fashion, giving rise to the notion of \emph{prism theories}, which we explore in more detail in~\ref{sec:PrismTheories}.

Noticing this, we could be tempted to conjecture that Theorem~\ref{thm:pairwise} could be extended to any $N$ in the na\"ive way, i.e.\ that the extreme states of a GPT with state space $\mathcal{S}_N^{\times l}$, the $l$-fold Cartesian product of the $N$-vertex simplex, form a mutually $N$-wise distinguishable set. However, we can quickly see that this is not the case, and that the relationship between simplex structure and the size of mutual distinguishability does not extend  beyond $N=2$. We show this with an example:

\begin{ex}
\label{ex:counter}
Call $\rho_1,\ldots, \rho_q$ the vectors of the canonical basis of $\R^q$, thought of as states in the classical GPT $\big(\R^q,\R_+^q, u\big)$ described in Example~\ref{ex class}. Explicitly, we will have $\rho_i = \left(0,\ldots, 1, \ldots 0\right)^\intercal$, where the single \tcb{non-zero} entry 
is in the $i^\text{th}$ position.
Then the extremal (pure) states of the $l$-fold product $\mathcal{S}_q^{\times l}$ are of the form $(\rho_{i_1},\ldots, \rho_{i_l})$, where $i_1,\ldots, i_l\in \{1,\ldots, q\}$. Consider the $3$ states
\begin{equation}
    \omega_1\coloneqq (\rho_1,\ldots,\rho_1)\, ,\qquad \omega_2\coloneqq (\rho_1,\rho_2,\rho_1,\ldots,\rho_1)\, ,\qquad \omega_3\coloneqq (\rho_2,\rho_1,\ldots,\rho_1)\, .
\end{equation}
Then $\omega_1,\omega_2,\omega_3$ are not jointly distinguishable. In fact, note that
\begin{equation}
    \omega_2+\omega_3 = (\rho_1+\rho_2,\rho_2+\rho_1,\rho_1,\ldots,\rho_1) \geq \omega_1\, .
\end{equation}
Thus, if $e_1\cdot \omega_2 = e_1\cdot \omega_3=0$, then also $e_1\cdot \omega_1=0$. In other words, there cannot be a measurement singling out $\omega_1$ from this triple of states.
\end{ex}

\subsection{Asymptotics in $m$ for fixed $N$}

Although no simple generalisations of the exact computation in Theorem~\ref{thm:pairwise} are available, we can obtain a general result that guarantees that for fixed $N$ and very large $m$, the scaling of the compression factor $\kappa(N,m)$ in $m$ is \emph{exactly the same} as that given by Theorem~\ref{thm:pairwise}. In other words, the scaling of $\kappa(N,m)$ in $m$ for a fixed $N$ does not depend on $N$. To prove this somewhat surprising result, 
we will make use of a probabilistic argument, while we leave open the task of finding a constructive proof of the result below.

\begin{thm} \label{thm:beyond_pairwise}
For all fixed integers $N\geq 2$, it holds that
\bb
\lim_{m\to\infty} \frac{\kappa(N,m)}{\kappa(2,m)} = \lim_{m\to\infty} \frac{\kappa(N,m)}{m/\log_2(m)} = 1\, .
\label{beyond_pairwise_limits}
\ee
\tcb{Equivalently,} for every fixed $N\geq 2$ we have that 
\bb
d(N,m)\leq m^{1 + o_N(1)} \qquad (m\to\infty)\, .
\label{dimension_bounds_beyond_pairwise}
\ee
\end{thm}

\tcb{
\begin{proof}
Clearly, if a set of GPT states is mutually $N$-wise distinguishable for some $N\geq 2$, it is also $2$-wise distinguishable (i.e.\ we have pairwise perfect distinguishability). Hence, $\kappa(N,m)\leq \kappa(2,m) = \frac{m}{\log_2(m+1)}$. 
%
Hence, inequalities~\eqref{beyond_pairwise_limits} and~\eqref{dimension_bounds_beyond_pairwise} are clearly equivalent, because
\bb
1\geq \frac{\kappa(N,m)}{\kappa(2,m)} = \frac{\kappa(N,m)}{m/\log_2(m)} = \frac{\log_2(m)}{\log_2 d(N,m)}\, ,
\ee
and the right-hand side tends to $1$ as $m\to\infty$ if and only if $d(N,m) \leq m^{1+o_N(1)}$. Now, to establish~\eqref{dimension_bounds_beyond_pairwise} we need to find, for fixed $N$ and very large $m$, an example of a GPT of dimension $m^{1+o_N(1)}$ that can accommodate a mutually $N$-wise distinguishable set of states of cardinality approximately $2^m$. To this end, consider the GPT with state space $\mathcal{S}_{q}^{\times l} = \mathcal{S}_{q(m)}^{\times l(m)}$ of Example~\ref{ex:counter}, where $q=q(m)$ and $l=l(m)$ are defined by
\bb
q(m) \coloneqq \floor{\big(\log_2(m)\big)^2}\,,\qquad l(m) \coloneqq \floor{\frac{2Nm}{\log_2\max\left\{\frac{2 q(m)}{N(N-1)}, 2\right\}}}\, .
\label{probabilistic_q_and_l}
\ee
Note that with these choices we have that
\bb
\log_2 \frac{N(N-1)}{2q(m)} = - \frac{2Nm}{l(m)}\, (1+o_N(1))\, .
\label{useful_relation_probabilistic_q_and_l}
\ee
Let us now draw states at random in an i.i.d.\ fashion from $\mathcal{S}_{q(m)}^{\times l(m)}$. Every state, of the form $\omega = (\rho_{i_1},\ldots, \rho_{i_l})$, is in turn constructed by drawing $i_1,\ldots, i_l\in \{1,\ldots, q(m)\}$ uniformly at random, again in an i.i.d.\ manner. We now ask ourselves: given $N$ random states $\omega^{(1)}, \ldots, \omega^{(N)}\in \mathcal{S}_{q(m)}^{\times l(m)}$, with $\omega^{(j)}=\big(\rho_{i_{1,j}}, \ldots, \rho_{i_{l,j}}\big)$, when are they perfectly distinguishable by looking only at the first components of each $\omega^{(j)}$, i.e.\ the states $\rho_{i_{1,j}}$, for $j=1,\ldots, N$? 
The answer to the above question is clear: whenever the first components of $\omega^{(1)}, \ldots, \omega^{(N)}$, i.e.\ the states $\rho_{i_{1,1}},\ldots, \rho_{i_{1,N}}$, are all different. This happens with probability
\bb
\pr\left\{ \text{$1^\text{st}$ component discriminates $\omega^{(1)}, \ldots, \omega^{(N)}$} \right\} = \prod_{k=1}^{N-1} \left(1-\frac{k}{q(m)}\right) ,
\ee
because $\prod_{k=1}^N \left(1-\frac{k}{q(m)}\right)$ is the probability that $N$ random numbers between $1$ and $q(m)$, in our construction $i_{1,1},\ldots, i_{1,N}$, 
are all different. Hence,
\bb
\pr\left\{ \text{$1^\text{st}$ component does not discriminate $\omega^{(1)}, \ldots, \omega^{(N)}$} \right\} = 1 - \prod_{k=1}^{N-1} \left(1-\frac{k}{q(m)}\right) .
\ee
Since we can look at any component of choice, there are $l(m)$ of them, and these are all independent,
\bb
&\pr\left\{ \text{no component discriminates $\omega^{(1)}, \ldots, \omega^{(N)}$} \right\} \\
&\qquad = \left(\pr\left\{ \text{$1^\text{st}$ component does not discriminate $\omega^{(1)}, \ldots, \omega^{(N)}$} \right\}\right)^{l(m)} \\
&\qquad = \left(1 - \prod_{k=1}^{N-1} \left(1-\frac{k}{q(m)}\right)\right)^{l(m)}\, .
\ee
So far we have only considered one $N$-tuple of states. If we draw a subset $\mathcal{M}\subset \mathcal{S}_{q(m)}^{\times l(m)}$ of $M= |\mathcal{M}| = 2^m$ states in total, there are $\binom{M}{N}$ distinct such $N$-tuples (up to re-ordering). Therefore, the probability that at least one of them is such that no component discriminates it is at most
\bb
&\pr \left(\bigcup_{\omega^{(1)},\ldots, \omega^{(N)}\in \mathcal{M}\ \text{distinct}} \left\{ \text{no component discriminates $\omega^{(1)}, \ldots, \omega^{(N)}$} \right\} \right) \\
&\qquad \leq \binom{M}{N}\, \pr \left\{ \text{no component discriminates $\omega^{(1)}, \ldots, \omega^{(N)}$} \right\} \\
&\qquad = \binom{M}{N} \left(1 - \prod_{k=1}^{N-1} \left(1-\frac{k}{q(m)}\right)\right)^{l(m)} .
\label{probabilistic_estimate}
\ee
As long as we can guarantee that the rightmost side of~\eqref{probabilistic_estimate} stays below $1$, 
we will know that there exists a choice of $\mathcal{M}$ such that for every distinct $\omega^{(1)},\ldots, \omega^{(N)}\in \mathcal{M}$, some component will discriminate them. Hence, we will have implicitly constructed a mutually $N$-wise distinguishable set $\mathcal{M}$ --- this is, of course, an instance of the celebrated probabilistic method~\cite{ERDOS, ALON-SPENCER}. And indeed, it is not difficult to show that the rightmost side of~\eqref{probabilistic_estimate} goes to $0$ as $m\to\infty$. Indeed, since $q(m)\tends{}{m\to\infty} \infty$ and $N$ is fixed one sees that
\bb
1 - \prod_{k=1}^{N-1} \left(1-\frac{k}{q(m)}\right) = (1+o_N(1)) \sum_{k=1}^{N-1} \frac{k}{q(m)} = (1+o_N(1))\, \frac{N(N-1)}{2q(m)}\, .
\ee
Thus,
\bb
&\log_2 \left(\binom{M}{N} \left(1 - \prod_{k=1}^{N-1} \left(1-\frac{k}{q(m)}\right)\right)^{l(m)}\right) \\
&\qquad \leq Nm \left\{ 1 + \frac{l(m)}{Nm} \log_2 \left(1 - \prod_{k=1}^{N-1} \left(1-\frac{k}{q(m)}\right)\right) \right\} \\
&\qquad = Nm \left\{ 1 + \frac{l(m)}{Nm} \log_2 \frac{N(N-1)}{2q(m)} + \frac{l(m)}{Nm} \log_2\big(1+o_N(1)\big) \right\} \\
&\qquad = Nm \left\{ 1 - 2 + o_N(1) + o_N\left(\frac{l(m)}{m}\right) \right\} \\
&\qquad = Nm \left\{ -1 + o_N(1) \right\} . 
\ee
where in the second line we used the crude approximation $\binom{M}{N}\leq M^N$, in the fourth we employed~\eqref{useful_relation_probabilistic_q_and_l}, and in the last we noted that $l(m)/m \tends{}{m\to\infty} 0$. 

This proves that for every fixed $N$ and all sufficiently large $m$, the GPT $\mathcal{S}_{q(m)}^{\times l(m)}$ can accommodate a mutually $N$-wise distinguishable set of states of cardinality $2^m$. Since the dimension of that GPT is $l(m)\big(q(m)-1\big) +1$, we deduce that
\bb
\frac{\kappa(N,m)}{m/\log_2(m)} &\geq \frac{\frac{m}{\log_2\left(l(m)(q(m)-1) +1\right)}}{m/\log_2(m)} \\
&= \frac{\log_2(m)}{\log_2\left(l(m)(q(m)-1) +1\right)} \\
&= \frac{\log_2(m)}{\log_2(m) + O_N\,\Big(\!\log_2 \!\big(\log_2 m\big)\Big)} \\
&\tends{}{m\to\infty} 1\, .
\label{probabilistic_ratio}
\ee
This concludes the proof.
%
%
\end{proof}
}


At this point, it is wise to pause for a moment our search for mutual $N$-distinguishable sets and ask ourselves a basic question: how do we decide whether a given a set of states is jointly perfectly distinguishable?

\section{Perfect distinguishability beyond pairwise: Numerical methods} \label{sec:beyondNum}

\subsection{Perfect distinguishability as a convex program}\label{sec:ConvexProgram}

We record here the simple observation that not only the question of perfect distinguishability, but actually the calculation of the minimal error probability in joint discrimination of a set of states in a given GPT is in fact a convex program~\cite{BOYD}. This is particularly interesting and useful, as in many situations arising naturally in applications the underlying cone admits an efficient description in terms of linear inequalities, or else in terms of inequalities in the L\"owner partial order, i.e.\ the one determined by positive semi-definiteness. The former is the case, for instance, for classical theories (Example~\ref{ex class}). A description in terms of positive semi-definite constraints, instead, can be formulated not only for quantum theory itself (Example~\ref{ex QM}), but also for several GPTs that are of great interest in entanglement theory~\cite{Horodecki-review}. Notable examples in this context include the theory of NPT entanglement~\cite{PeresPPT, Martin-exact-PPT, irreversibility-PPT, Xin-exact-PPT, PPT-high-SN} and that of extendibility~\cite{complete-extendibility, Bhat16, Kaur2018, Kaur2018-PRA, extendibility}.

\begin{lemma} \label{convex_program_lemma}
Let $(V,C,u)$ be a $d$-dimensional GPT with state space $\Omega$. Given states $\{\omega_i\}_{i=1}^N \subset \Omega$ and a priori probabilities $\{p_i\}_{i=1}^N$, the maximal success probability in the associated task of state discrimination is given by the convex program
\bb
\begin{array}{llll}
P_s^{\max}\left(\{p_i,\omega_i\}_{i=1}^N\right) & = & \max & \sum_{i=1}^N p_i\, e_i \cdot \omega_i \\[1ex]
&& \mathrm{s.t.} & e_1,\ldots, e_N \in C^*,\ \sum_{i=1}^N e_i = u\, .
\end{array}
\label{p_success_discrimination}
\ee
If $C$ is polyhedral with $M$ extremal rays, i.e.\ if there exist finitely many $v_1,\ldots, v_M\in V$ such that $C = \left\{\sum_{j=1}^M a_j v_j:\,  a_j\geq 0\ \, \forall\, j\right\}$, then~\eqref{p_success_discrimination} can be rephrased as a linear program, namely,
\bb
\arraycolsep=1.6pt 
\begin{array}{llll}
P_s^{\max}\left(\{p_i,\omega_i\}_{i=1}^N\right) & = & \max & \sum_{i=1}^N p_i\, e_i \cdot \omega_i \\[1ex]
&& \mathrm{s.t.} & e_i\cdot v_j\geq 0\quad \forall\ i=1,\ldots, N,\quad \forall\ j=1,\ldots, M,\ \sum_{i=1}^N e_i = u\, .
\end{array}
\label{p_success_discrimination_polyhedral}
\ee
The above program can be solved efficiently, in time $O\left( d(d+M)^{3/2} N^{5/2}\right)$.
\end{lemma}

\begin{proof}
The most general state discrimination procedure consists of making a measurement $(e_i)_{i=1,\ldots, N}$, and guessing the unknown state to be $\omega_i$ upon having obtained outcome $i$. The average probability of success of this strategy is precisely $\sum_i p_i\, e_i\cdot \omega_i$. The constraints in~\eqref{p_success_discrimination} are those required to make sure that $(e_i)_{i=1,\ldots, N}$ is in fact a valid measurement in the GPT $(V,C,u)$.

If $C$ is polyhedral with $M$ extremal rays spanned by vectors $v_1,\ldots, v_M$, then naturally $e\in V^*$ satisfies that $e\in C^*$ if and only if $e\cdot v_j\geq 0$ for all $j=1,\ldots,M$. In this way one derives~\eqref{p_success_discrimination_polyhedral} from~\eqref{p_success_discrimination}. Finally, the estimates on the efficiency of the linear program solution are taken from the work by Vaidya~\cite{Vaidya1989}. To make the comparison precise, note that in our case $n = d(N-1) \sim dN$ is the number of real variables\footnote{We have $d$ variables for each of the vectors $e_1,\ldots, e_{N-1}$ living in a $d$-dimensional space $V^*$. Note that $e_N$ is uniquely determined by the normalisation condition $\sum_i e_i=u$.} and $m=NM$ is the number of constraints.
\end{proof}

Based on the above result, we can state its implications for the problem of perfect discrimination, which is of interest here:

\begin{cor}
Let $(V,C,u)$ be a $d$-dimensional GPT with state space $\Omega$. Deciding whether the states $\{\omega_i\}_{i=1}^N \subset \Omega$ are perfectly distinguishable is a convex feasibility problem~\cite{BOYD}:
\bb
\begin{array}{ll}
\mathrm{find} & e_1,\ldots, e_N \in V^* \\
\mathrm{s.t.} & e_1,\ldots, e_N \in C^*,\quad \sum_{i=1}^N e_i = u,\quad e_i\cdot \omega_i=1\ \ \forall\ i=1,\ldots, N\, .
\end{array}
\ee
If $C$ is polyhedral with $M$ extremal rays then the above program becomes linear, and can be solved in time at most $O\left( d(d+M)^{3/2} N^{5/2}\right)$.
\end{cor}

\subsection{Restricting the search}\label{sec:RestrictingTheSearch}

Consider for simplicity a GPT $(V,C,u)$ whose cone is polyhedral. Thanks to Lemma~\ref{convex_program_lemma}, we know that whether a given set of states $\{\omega_i\}_i$ can be discriminated perfectly can be decided efficiently. But how do we start searching for a maximal mutually $N$-wise distinguishable set of states?
Before we proceed to answer this, we first show that the search can be restricted to pure states, i.e.\ to extremal points of the state space.

\begin{lemma}
Let $(V,C,u)$ be a $d$-dimensional GPT with state space $\Omega$. Given some $N\geq 2$, a mutually $N$-wise distinguishable set can be searched among pure states, i.e.\ extremal points of $\Omega$.
\end{lemma}

\begin{proof}
Let us assume that a mutually $N$-wise distinguishable set $\mathcal{S}\subseteq \Omega$ has been found. Every $\omega\in \mathcal{S}$ will admit a (not necessarily unique) decomposition of the form $\omega = \sum_{i=1}^d p_i^\omega \varphi_i^\omega$, where $\varphi_i^\omega \in \Omega$ are pure states. Let us pick $i$ such that $p_i^\omega>0$, and consider the associated pure state $\varphi_i^\omega$. Repeating this procedure for every $\omega\in \mathcal{S}$, we can form a set of pure states $\mathcal{S}' = \{ \varphi_i^\omega:\, \omega\in \mathcal{S},\, p_i^\omega>0 \}$.

We claim that also $\mathcal{S}'$ is mutually $N$-wise distinguishable. To see why, pick some pure states $\varphi_{i_1}^{\omega_1},\ldots, \varphi_{i_N}^{\omega_N}\in \mathcal{S}'$, and consider the corresponding states $\omega_1,\ldots, \omega_N\in \mathcal{S}$. Let $(e_j)_{j=1,\ldots, N}$ be the measurement that achieves perfect discrimination of the set $\{\omega_j\}_j$, i.e.\ such that $1= e_j\cdot \omega_j=\sum_i p_i^{\omega_i} e_j \cdot \varphi_i^{\omega_j}$ for all $j$. We immediately deduce that $e_j\cdot \varphi_i^{\omega_j}=1$ for all $i$ and $j$ such that $p_i^{\omega_j}>0$, and in particular $e_j\cdot \varphi_{i_j}^{\omega_j}=1$ for all $j$. This implies that the states $\varphi_{i_1}^{\omega_1},\ldots, \varphi_{i_N}^{\omega_N}$ are perfectly distinguishable by means of the measurement $(e_j)_{j=1,\ldots, N}$.
\end{proof}

We could wonder whether a similar restriction applies to the measurements as well, i.e.\ whether it suffices to restrict the search to extremal effects. After all, if $e\cdot \omega=1$ and $e = \sum_i p_i^e f_i$ with $f_i$ extremal effects, it follows that $f_i \cdot \omega=1$ whenever $p_i>0$; we could therefore imagine to replace $e$ with any $f_i$ such that $p_i>0$. The reason why this does not work, however, is that doing so in general alters the sum of all the effects, which needs to be equal to the order unit. This means that in general restricting to extremal effects is not guaranteed to yield all possible feasible measurements. We construct an example to demonstrate this in~\ref{app:restricting_search}.

Nonetheless, the fact that we can restrict ourselves to the finite set of pure states makes our search for the largest $N$-wise mutually distinguishable set of states {much} easier to approach. Restricting ourselves to extremal effects would have been useful in that it would 
have enabled us to simplify the search for distinguishing measurements, but the convex approach described in Section~\ref{sec:ConvexProgram} serves perfectly well to that purpose. The restriction to pure states, on the other hand, means that the next component of our search can take place on the terrain of a finite, rather than infinite, set.

\subsection{Finding the largest $N$-wise mutually distinguishable set of states}\label{sec:SearchForMutually}
Our search can be split into two distinct steps:
\begin{itemize}
    \item \textbf{Joint distinguishability:} For a given GPT with a state space $\Omega$, discover all subsets  of $\ext(\Omega)$ which are {$N$-wise} distinguishable. For example, if we had $N=3$, we would be finding all triples of pure states which were distinguishable by a tripartite measurement. We call such sets $\theta^i$, and the set of such sets $\Theta = \{\theta^i\}$.
    \item \textbf{Mutual distinguishability:} Find the largest set $\Phi\subseteq\ext(\Omega)$ such that every subset $\phi \subset \Phi$ of cardinality $|\phi| = N$ is also an element of $\Theta$. This  means that any $N$ size subset of $\Phi$ is $N$-wise distinguishable --- or, equivalently, that $\Phi$ is $N$\emph{-wise mutually distinguishable}.
\end{itemize}

The first of these steps can be straightforwardly achieved using the methods described in Section~\ref{sec:ConvexProgram}.
%
%
Once the set of jointly distinguishable sets $\Theta$ is in hand, we can proceed to finding the largest $N$-wise mutually distinguishable set $\Phi$. Given that we know the elements of $\Theta$, we know all groups of states which are $N$-wise distinguishable. We can think of this relationship between states --- that of being $N$-wise jointly distinguishable ---  as a connection between them. In fact, we can take this logic literally; we can construct a hypergraph overlaying our state space. Formally, recall that an (undirected) hypergraph is a pair $(V,E)$, where $V$ is a (finite) set of so-called nodes, and $E$ is a subset of the power set of $V$, i.e.\ a collection of subsets of $V$. We refer to the elements of $E$ as hyperedges, and to $E$ itself as the hyperedge set. A hypergraph is called $N$-regular if each hyperedge has cardinality precisely $N$.

In the hypergraph we construct, the vertices correspond to the states we are considering (typically the pure states of the theory), and the hyperedges are all the subsets of $N$ states that are perfectly distinguishable. 
Note that if $N=2$ then every edge connects two vertices, yielding an ordinary graph. 

\begin{Def}[Distinguishablity hypergraphs] \label{def:hypergraph}
Given some integer $N\geq 2$ and a GPT with state space $\Omega$ and finitely many pure states, i.e.\ such that $|\ext (\Omega)|<\infty$, the $N$-distinguishability hypergraph of $\Omega$, denoted \tcb{$\mathcal{G}(\Omega; N)=\big(\ext(\Omega),\Theta\big)$, is the $N$-regular hypergraph with node set $\ext(\Omega)$ and set of hyperedges $\Theta$} given by all subsets of $\ext(\Omega)$ of cardinality $N$ which are perfectly distinguishable in $\Omega$ according to Definition~\ref{def:perfectly_distinguishable}.
\end{Def}

Above, we described the task of finding $\Phi$ as follows:

\begin{quote}
\emph{Find the largest set $\Phi\subseteq\ext(\Omega)$ such that every subset $\phi \subset \Phi$ of cardinality $|\phi| = N$ is also an element of $\Theta$.}
\end{quote}

With the graph-theoretic view of our problem in mind, we can re-formulate this problem as follows:

\begin{quote}
\emph{What is the largest sub-graph $\mathcal{G}'$ of $\mathcal{G}(\Omega; N)$ which is $N$-complete, in the sense that every subset of nodes of $\mathcal{G}'$ of cardinality $N$ is a hyperedge?}
\end{quote}

This being a particular phrasing of the well-known \emph{maximum clique problem}. Or, to be more precise; in our case we seek the maximum $N$-clique on an $N$-regular hypergraph. On ordinary graphs (ordinary in the sense that they are not hypergraphs), the problem is well studied~\cite{Vassilevska2009,Cazals2008, Sun2020, Carraghan1990, Segundo2011}, and algorithms are known both for exact solutions, and for faster, inexact solutions --- a review appears in~\cite{Wu2015}. This problem is known to be \textbf{\textsc{NP}}-complete~\cite{Karp1972}.

In the case of hypergraphs, however, less is known. \tcb{The problem can be tackled by adapting an existing algorithm called \texttt{hClique}~\cite{Torres-Jimenez2017}. In our notation, the procedure works by examining each edge in turn, and finding the largest clique $Q$ branching out from the nodes on that edge. In order to do this, we begin with an edge $\theta$, and set the initial clique $Q$ to the nodes connected by that edge $\theta$.  We then examine the set of nodes not included in $\theta$, $\Omega'=\ext(\Omega) \setminus \theta$. Then, for each $\omega\in \Omega'$, we check if $\omega$ is fully connected to $Q$. If it is, then it can be added to $Q$, and the clique can grow.}

\begin{Def}[Fully connected cliques] \label{def:fully_connected}
\tcb{Let $\Omega$ be the state space for a GPT $(V,C,u)$, and assume that $\Omega$ has finitely many pure states, i.e.\ that $|\ext(\Omega)|<\infty$. 
Let $\mathcal{G}(\Omega; N) = \big(\ext(\Omega),\Theta\big)$ be the $N$-distinguishability hypergraph of $\Omega$, as per Definition~\ref{def:hypergraph}. Let $\theta\in \Theta$ be a hyperedge of $\mathcal{G}(\Omega; N)$. A node $\omega\in \ext(\Omega)\setminus \theta$ is fully $N$-connected to $\theta$ if for all $a\in \theta$ it holds that $(\theta\setminus a)\cup \{\omega\}\in \Theta$.
}
\end{Def}


This process will discover the largest clique which can be built out from each edge; this is the set of maximal cliques. The largest of these will be the \emph{maximum clique}, our object of interest.


\begin{figure}
    \centering
    \includegraphics[width=\linewidth]{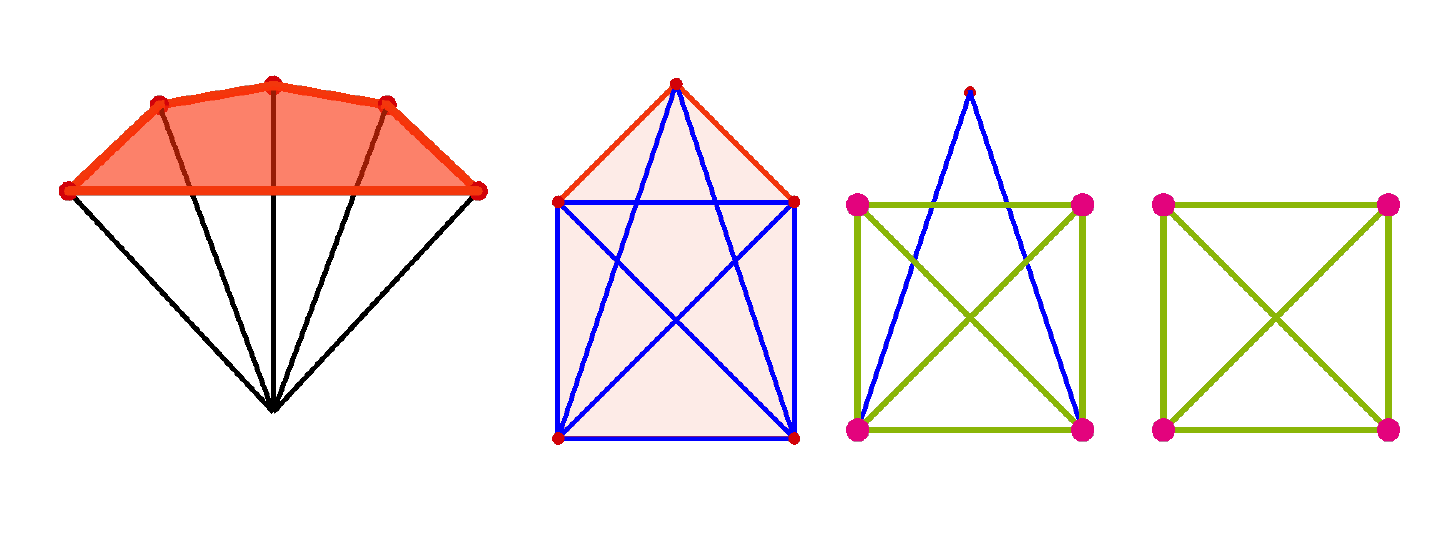}
    \caption{Illustration of how we would isolate the maximum $N$-clique for a simple state space. From left to right: The first image the state space of a hypothetical GPT. In the next image, we discard the third dimension and connect the states which are pairwise distinguishable using blue lines. The third image then isolates those states which form a clique: all four of the larger, pink vertices are connected to one another by blue edges, meaning they are \emph{mutually pairwise distinguishable}. The top state is only connected to two of these, so is excluded from the clique. Since no node can be added to this clique, it is maximal.  In this case this clique is also the maximum clique. The rightmost and final image isolates this maximum clique as a graph of four states.}
    \label{fig:MaxCliqueRowImage}
\end{figure}

If we take $N=2$, we have the simpler problem of finding the maximum clique on a (non-hyper) graph. Note that this is the case even for high-dimensional state spaces, because the dimensionality of the hypergraph described in Definition~\ref{def:hypergraph} depends only upon the number of states connected by each edge (perfectly distinguishable through a single measurement), not upon the dimension of $V$ itself. The case depicted in Figure~\ref{fig:MaxCliqueRowImage} is two dimensional in two senses: the original state space occupies a two-dimensional surface embedded in $\R^3$, and the graph formed from it can be represented on the plane, since each edge is a line.

As stated above, the problem of discovering the largest set of $N$-wise mutually distinguishable states for a given GPT --- equivalent to finding the maximum $N$-clique for an undirected hypergraph --- may not be amenable to a closed analytical solution in general. Even though we do not yet know which theories would be optimal for associative memories in the case that $N>2$, our methods in this Section reveal an exact approach for probing candidate theories, in any dimension and for any $N$.

\section{Discussion}\label{sec:Discussion}

In this paper we discussed a simple model of associative memory, in the form of a GPT system capable of being in any one of $2^m$ states in such a way that any $N$ of them are perfectly distinguishable. When $N=2$, we could characterise precisely the GPTs performing optimally at this task: they are theories whose state space is shaped as an $m$-dimensional hypercube. We proved in Theorem~\ref{thm:pairwise} that such theories outperform classical and quantum theories exponentially, in the sense that they have dimension $d(2,m)=m+1$, while any classical or quantum system with the same properties needs to have dimension $O(2^m)$.
We extended our analysis to the asymptotic case of arbitrary fixed $N$ and very large $m$, proving in Theorem~\ref{thm:beyond_pairwise} that there exist GPTs with dimension still scaling effectively linearly with $m$, $d(N,m) \leq m^{1+o_N(1)}$ (as $m\to\infty$), for every $N\geq 2$. This  means that, in such a ``big data'' scenario, the exponential improvement enabled by GPTs over classical and quantum theories is  independent of $N$; in other words, \emph{there is plenty of room in the post-quantum world.} \tcb{Following the completion of this paper, further developments of these and related ideas have been presented in recent works~\cite{Weiner2023,Naszodi2023}.}

Though we were not able to generalise our optimality construction -- we do not know, for any given value of $N>2$, what the optimal GPT would be -- we have shown that there exists a reliable and computationally tractable method for discovering the memory capacity of theories for any $N$.
To recap the method, we first showed that we can restrict the search to $N$-sized  subsets of \emph{pure} jointly distinguishable states. The set of such subsets of jointly distinguishable states can be thought of as a connection hypergraph overlaying the set  of pure states of the GPT. Our search for the largest $N$-wise mutually distinguishable set  thus becomes equivalent to the search for the maximum clique on this hypergraph, which can be performed with deterministic success~\cite{Torres-Jimenez2017}

We hope this work could inspire further research into the cognitive abilities of intelligent agents in generalised probabilistic theories as well as their interplay with classical and quantum learning models~\cite{transformer,ultimatebrain}.

\section*{Acknowledgements}

LL thanks Guillaume Aubrun \tcb{and Mih\'{a}ly Weiner} for inspiring discussions on this topic. He is indebted to Guillaume Aubrun as well as to Boaz Slomka for bringing to his attention the paper by Danzer and Gr\"unbaum~\cite{Danzer1962}. LL acknowledges support from the Alexander von Humboldt Foundation. DG and GA thank Paul Knott for illuminating discussions, and acknowledge support from the Foundational Questions Institute (FQXi) under the Intelligence in the Physical World Programme (Grant No. RFP-IPW1907).

\section*{References}

\bibliographystyle{iopart-num-g}
\bibliography{biblio,library}

\providecommand{\newblock}{}
\begin{thebibliography}{10}
\expandafter\ifx\csname url\endcsname\relax
  \def\url#1{{\tt #1}}\fi
\expandafter\ifx\csname urlprefix\endcsname\relax\def\urlprefix{URL }\fi
\providecommand{\eprint}[2][]{\url{#2}}

\bibitem{Ambainis1998}
Ambainis A, Nayak A, Ta-Shma A and Vazirani U 1999 {\em Proc. 31st ACM Symp. on
  Theory of Computing\/} STOC '99 (New York, NY, USA) pp 376--383 ISBN
  1581130678

\bibitem{VerSteeg2009}
Ver~Steeg G and Wehner S 2009 {\em Quantum Info. Comput.\/} {\bf 9} 801--832
  ISSN 1533-7146

\bibitem{Hopfield1982}
Hopfield J~J 1982 {\em Proc. Natl. Acad. Sci. U.S.A.\/} {\bf 79} 2554--2558
  ISSN 0027-8424

\bibitem{implausible}
van Dam W 2013 {\em Nat. Comput.\/} {\bf 12} 9--12

\bibitem{PVV}
Linden N, Popescu S, Short A~J and Winter A 2007 {\em Phys. Rev. Lett.\/} {\bf
  99}(18) 180502

\bibitem{ultimate}
Lami L, Palazuelos C and Winter A 2018 {\em Commun. Math. Phys.\/} {\bf 361}
  661--708

\bibitem{XOR}
Aubrun G, Lami L, Palazuelos C, Szarek S~J and Winter A 2020 {\em Commun. Math.
  Phys.\/} {\bf 375} 679--724

\bibitem{cones-2}
Aubrun G, Lami L, Palazuelos C and Pl\'avala M 2021 {\em Geom. Funct. Anal.\/}
  {\bf 31} 181--205

\bibitem{cones-3}
Aubrun G, Lami L, Palazuelos C and Pl\'{a}vala M 2022 {\em Phys. Rev. Lett.\/}
  {\bf 128}(16) 160402 editors' Suggestion

\bibitem{Barrett-original}
Barrett J 2007 {\em Phys. Rev. A\/} {\bf 75} 032304

\bibitem{PR-boxes}
Popescu S and Rohrlich D 1994 {\em Found. Phys.\/} {\bf 24} 379--385

\bibitem{Popescu1997}
Popescu S and Rorlich D 1997 {\em Phys. Rev. A\/} {\bf 56} R3319

\bibitem{lamiatesi}
Lami L 2017 {\em Non-classical correlations in quantum mechanics and beyond\/}
  Ph.D. thesis Universitat Aut\`onoma de Barcelona preprint arXiv:1803.02902

\bibitem{Mueller2021}
M\"{u}ller M 2021 {\em SciPost Phys. Lect. Notes\/}

\bibitem{Plavala2021}
Pl\'{a}vala M 2021 {\em Preprint arXiv:2103.07469\/}

\bibitem{Danzer1962}
Danzer L and Gr{\"u}nbaum B 1962 {\em Math. Zeitschrift\/} {\bf 79} 95--99

\bibitem{Arai2019}
Arai H, Yoshida Y and Hayashi M 2019 {\em J. Phys. A\/} {\bf 52} 465304

\bibitem{Yoshida2020}
Yoshida Y, Arai H and Hayashi M 2020 {\em Phys. Rev. Lett.\/} {\bf 125}(15)
  150402

\bibitem{Chiribella-pur}
Chiribella G, D'Ariano G~M and Perinotti P 2010 {\em Phys. Rev. A\/} {\bf
  81}(6) 062348

\bibitem{Lee2015a}
Lee S~Y, Lee C~W, Nha H and Kaszlikowski D 2015 {\em J. Opt. Soc. Am. B\/} {\bf
  32} 1186--1192

\bibitem{Lee2016a}
Lee M 2016 {\em {Bounds on computation from physical principles}\/} Ph.D.
  thesis Oxford University

\bibitem{mueller2020}
Mueller M~P 2020 {\em Preprint arXiv:2011.01286\/}

\bibitem{LUDWIG}
Ludwig G 1985 {\em An Axiomatic Basis for Quantum Mechanics: Derivation of
  {H}ilbert space structure\/} vol~1 (Springer-Verlag)

\bibitem{FOUNDATIONS}
Hartk{\"a}mper A and Neumann H 1974 {\em Foundations of Quantum Mechanics and
  Ordered Linear Spaces: Advanced Study Institute held in {M}arburg 1973\/}
  (Springer Berlin Heidelberg)

\bibitem{Davies-1970}
Davies E~B and Lewis J~T 1970 {\em Commun. Math. Phys.\/} {\bf 17} 239--260

\bibitem{no-restriction}
Janotta P and Lal R 2013 {\em Phys. Rev. A\/} {\bf 87}(5) 052131

\bibitem{telep-in-GPT}
Barnum H, Barrett J, Leifer M and Wilce A 2012 {\em Proc. Sympos. Appl.
  Math.\/} vol~71 pp 25--48

\bibitem{Shahandeh2021}
Shahandeh F 2021 {\em PRX Quantum\/} {\bf 2} 1 ISSN 2691-3399
  (\textit{Preprint} \eprint{1911.11059})

\bibitem{Schmid2020}
Schmid D, Selby J~H, Wolfe E, Kunjwal R and Spekkens R~W 2021 {\em PRX
  Quantum\/} {\bf 2}(1) 010331

\bibitem{DAriano2020}
D'Ariano G~M, Erba M and Perinotti P 2020 {\em Phys. Rev. A\/} {\bf 101}(4)
  042118

\bibitem{cones-1}
Aubrun G, Lami L and Palazuelos C 2019 {\em Preprint arXiv:1910.04745\/}

\bibitem{Janotta2011}
Janotta P, Gogolin C, Barrett J and Brunner N 2011 {\em New J. Phys.\/} {\bf
  13} ISSN 13672630 (\textit{Preprint} \eprint{1012.1215})

\bibitem{Short2010}
Short A~J and Wehner S 2010 {\em New J. Phys.\/} {\bf 12} ISSN 13672630
  (\textit{Preprint} \eprint{0909.4801})

\bibitem{Massar2014a}
Massar S and Patra M~K 2014 {\em Physical Review A - Atomic, Molecular, and
  Optical Physics\/} {\bf 89} 1--8 ISSN 10941622 (\textit{Preprint}
  \eprint{1403.2509})

\bibitem{Heinosaari2019}
Heinosaari T, Lepp{\"{a}}j{\"{a}}rvi L and Pl{\'{a}}vala M 2019 {\em Quantum\/}
  {\bf 3} 1--39 ISSN 2521327X (\textit{Preprint} \eprint{1808.07376})

\bibitem{Kobayshi2017}
Kobayshi M 2017 {\em J. Math. Phys.\/} {\bf 58} ISSN 00222488

\bibitem{Pfister2013}
Pfister C and Wehner S 2013 {\em Nat. Commun.\/} {\bf 4} 1--9 ISSN 20411723

\bibitem{Masanes2014}
Masanes L, M{\"{u}}ller M~P, P{\'{e}}rez-Garc{\'{i}}a D and Augusiak R 2014
  {\em J. Math. Phys.\/} {\bf 55} ISSN 00222488 (\textit{Preprint}
  \eprint{1111.4060})

\bibitem{Krumm2019}
Krumm M and M{\"{u}}ller M~P 2019 {\em NPJ Quantum Information\/} {\bf 5} ISSN
  20566387 (\textit{Preprint} \eprint{1804.05736})

\bibitem{KLEE}
Klee V 1960 {\em Unsolved problems in intuitive geometry\/} (Hektographiert,
  Seattle)

\bibitem{THEBOOK}
Aigner M and Ziegler G~M 2010 {\em Proofs from {T}he {B}ook\/} 4th ed
  (Springer-Verlag, Berlin) ISBN 978-3-642-00855-9

\bibitem{ERDOS}
Erdos P and Spencer J 1974 {\em AMC\/} {\bf 10} 12

\bibitem{ALON-SPENCER}
Alon N and Spencer J~H 2016 {\em The probabilistic method\/} (John Wiley \&
  Sons)

\bibitem{BOYD}
Boyd S~P and Vandenberghe L 2004 {\em Convex Optimization\/} Berichte {\"u}ber
  verteilte messysteme (Cambridge University Press)

\bibitem{Horodecki-review}
Horodecki R, Horodecki P, Horodecki M and Horodecki K 2009 {\em Rev. Mod.
  Phys.\/} {\bf 81}(2) 865--942

\bibitem{PeresPPT}
Peres A 1996 {\em Phys. Rev. Lett.\/} {\bf 77}(8) 1413--1415

\bibitem{Martin-exact-PPT}
Audenaert K, Plenio M~B and Eisert J 2003 {\em Phys. Rev. Lett.\/} {\bf 90}(2)
  027901

\bibitem{irreversibility-PPT}
Wang X and Duan R 2017 {\em Phys. Rev. Lett.\/} {\bf 119}(18) 180506

\bibitem{Xin-exact-PPT}
Wang X and Wilde M~M 2020 {\em Phys. Rev. Lett.\/} {\bf 125}(4) 040502

\bibitem{PPT-high-SN}
Huber M, Lami L, Lancien C and M\"uller-Hermes A 2018 {\em Phys. Rev. Lett.\/}
  {\bf 121}(20) 200503

\bibitem{complete-extendibility}
Doherty A~C, Parrilo P~A and Spedalieri F~M 2004 {\em Phys. Rev. A\/} {\bf
  69}(2) 022308

\bibitem{Bhat16}
Rajarama~Bhat B~V, Parthasarathy K~R and Sengupta R 2017 {\em Rev. Math.
  Phys.\/} {\bf 29} 1750012

\bibitem{Kaur2018}
Kaur E, Das S, Wilde M~M and Winter A 2019 {\em Phys. Rev. Lett.\/} {\bf
  123}(7) 070502

\bibitem{Kaur2018-PRA}
Kaur E, Das S, Wilde M~M and Winter A 2021 {\em Phys. Rev. A\/} {\bf 104}(2)
  022401

\bibitem{extendibility}
Lami L, Khatri S, Adesso G and Wilde M~M 2019 {\em Phys. Rev. Lett.\/} {\bf
  123}(5) 050501

\bibitem{Vaidya1989}
Vaidya P~M 1989 {\em Proc. 30th Annual Symp. Found. Computer Science\/} SFCS
  '89 (USA: IEEE Computer Society) pp 332--337 ISBN 0818619821

\bibitem{Vassilevska2009}
Vassilevska V 2009 {\em Information Processing Letters\/} {\bf 109} 254--257
  ISSN 00200190

\bibitem{Cazals2008}
Cazals F and Karande C 2008 {\em Theoretical Computer Science\/} {\bf 407}
  564--568 ISSN 03043975

\bibitem{Sun2020}
Sun B, Danisch M, Chan T~H and Sozio M 2020 {\em Proceedings of the VLDB
  Endowment\/} {\bf 13} 1628--1640 ISSN 21508097

\bibitem{Carraghan1990}
Carraghan R and Pardalos P~M 1990 {\em Operations Research Letters\/} {\bf 9}
  375--382 ISSN 01676377

\bibitem{Segundo2011}
Segundo P~S, Rodr{\'{i}}guez-Losada D and Jim{\'{e}}nez A 2011 {\em Computers
  and Operations Research\/} {\bf 38} 571--581 ISSN 03050548

\bibitem{Wu2015}
Wu Q and Hao J~K 2015 {\em European Journal of Operational Research\/} {\bf
  242} 693--709 ISSN 03772217

\bibitem{Karp1972}
Karp R~M 1972 {\em Complexity of Computer Computations\/}  85--103

\bibitem{Torres-Jimenez2017}
Torres-Jimenez J, Perez-Torres J~C and Maldonado-Martinez G 2017 {\em Discrete
  Mathematics, Algorithms and Applications\/} {\bf 9} ISSN 17938317

\bibitem{Weiner2023}
Weiner M 2023 {\em Preprint arXiv:2301.06553\/}

\bibitem{Naszodi2023}
Nasz\'{o}di M, Szil\'{a}gyi Z and Weiner M 2023 {\em Preprint
  arXiv:2307.16857\/}

\bibitem{transformer}
Vaswani A, Shazeer N, Parmar N, Uszkoreit J, Jones L, Gomez A~N, Kaiser {\L}
  and Polosukhin I 2017 {\em Advances in neural information processing
  systems\/} pp 6000--6010

\bibitem{ultimatebrain}
Adesso G 2023 {\em AI Magazine\/} {\bf 44} 328--342

\bibitem{ROCKAFELLAR}
Rockafellar R~T 1970 {\em Convex analysis\/} (Princeton University Press,
  Princeton, N.J.)

\bibitem{Lang1986}
Lang R 1986 {\em Arch. Math.\/} {\bf 47} 90--92

\bibitem{Tabia2013}
Tabia G~N~M and Appleby D~M 2013 {\em Phys. Rev. A\/} {\bf 88} 1--8 ISSN
  10502947 (\textit{Preprint} \eprint{1304.8075})

\end{thebibliography}

\appendix

\section{Proof of the Danzer--Gr\"unbaum Theorem}\label{appendix:DG}

\subsection*{Preliminaries: Minkowski addition}

Before delving into the proof, we need to fix some terminology. A \emph{convex body} in the Euclidean space $\R^n$ is a compact convex subset $A \subset \R^n$ with non-empty interior, in formula $\inter(A)\neq \emptyset$. We say that two convex bodies $A,B\subset \R^n$ \emph{touch each other} if $A\cap B \neq \emptyset$ but $\inter(A)\cap \inter(B)=\emptyset$, which corresponds to the intuitive notion of two solids touching only at their surfaces.

Two sets $A,B\subseteq \R^n$ can be added together via the \emph{Minkowski addition}, defined by
\begin{equation}
    A+B\coloneqq \left\{a+b:\, a\in A,\, b\in B\right\} .
\end{equation}
In what follows, for some $x\in \R^n$ we will often write $x+A$ instead of $\{x\}+A$. We can also multiply a given set by any real number $\lambda\in \R$, by setting
\begin{equation}
    \lambda A \coloneqq \left\{\lambda a:\, a\in A\right\} .
\end{equation}
Naturally, the Minkowski difference between two sets $A,B\subseteq \R^n$ can now be constructed as $A-B\coloneqq A + (-B)$. If $A$ and $B$ are convex then also $A+B$ and $\lambda A$ are such. If they are convex bodies and $\lambda\neq 0$, then also $A+B$ and $\lambda A$ are convex bodies. A special type of Minkowski addition is the \emph{Minkowski symmetrisation}. For $A\subseteq \R^n$, this is defined by
\begin{equation}
    \wtildea{A} \coloneqq \frac{A-A}{2}\, .
\end{equation}
Clearly, if $A$ is a convex body then so is $\wtildea{A}$. In what follows we will need the following standard lemma, whose proof is included only for the sake of completeness (it follows e.g.\ from~\cite[Corollary~6.6.2]{ROCKAFELLAR}).


\begin{lemma} \label{interior Minkowski symm lemma}
Let $A\subset \R^n$ be a convex body. Then
\bb
\wtildea{\inter A} = \inter \wtildea{A}\, .
\ee
\end{lemma}

\begin{proof}
We start by showing that $\wtildea{\inter A} \subseteq \inter \wtildea{A}$. First, note that $\wtildea{\inter A} \subseteq \wtildea{A}$, simply because $\inter(A)\subseteq A$. Second, observe that $\wtildea{\inter A}$ is open. To show this, pick some $a\in \wtildea{\inter A}$, so that $a=\frac{b-c}{2}$ with $b,c\in \inter A$. Let $\epsilon>0$ be such that $\|\delta\|<\epsilon$ implies that $b+\delta,\,c+\delta\in \inter A$, where $\|\cdot\|$ denotes the Euclidean norm. Then as long as $\|\delta\|\leq \epsilon$ we also have that $a+\delta = \frac{(b+\delta) - (c-\delta)}{2}\in \wtildea{\inter A}$. This confirms that $\wtildea{\inter A}$ is indeed open. Since the interior of a set is nothing but its largest open subset, from this and the inclusion $\wtildea{\inter A} \subseteq \wtildea{A}$ we deduce that $\wtildea{\inter A} \subseteq \inter \wtildea{A}$.

For the other inclusion, take $a\in \inter \wtildea{A}$, and some sufficiently small $\epsilon>0$ such that $\frac{a}{1-\epsilon} = \frac{b-c}{2} \in \wtildea{A}$, where $b,c\in A$ (note that the left-hand side converges to $a$ as $\epsilon\to 0^+$ and is thus eventually in $\wtildea{A}$). Consider a point $p\in \inter A$; we now claim that $(1-\epsilon)b+\epsilon p,\, (1-\epsilon) c +\epsilon p\in \inter A$ for all $0<\epsilon<1$. To see this geometrically intuitive fact, fix $\epsilon>0$ and pick $\eta>0$ such that $\|\delta\|\leq \eta$ implies that $p+\delta\in A$. Then as soon as $\left\|\delta'\right\|\leq \epsilon\eta$ we have that for example $(1-\epsilon)b+\epsilon p + \delta' = (1-\epsilon)b+\epsilon (p + \delta)\in A$, where $\delta\coloneqq \delta'/\epsilon$. This proves that $(1-\epsilon)b+\epsilon p,\, (1-\epsilon) c +\epsilon p\in \inter A$, as claimed. Now,
\bb
a = (1-\epsilon)\frac{b-c}{2} = \frac{\left( (1-\epsilon) b+\epsilon p\right) - \left( (1-\epsilon) c + \epsilon p\right)}{2} \in \frac12 \left( \inter A - \inter A\right) = \wtildea{\inter A}\, ,
\ee
concluding the proof.
\end{proof}

\subsection*{The proof}

We are now ready to present Danzer and Gr\"unbaum's argument~\cite{Danzer1962}, in a slightly simplified form.

\begin{proof}[Proof of Theorem~\ref{thm:DG} and therefore of Theorem~\ref{thm:pairwise}]
For a positive integer $n$, some finite subset $X\subset \R^n$, and a convex body $A\subset \R^n$, we define the following properties:
\begin{itemize}
    \item $P(n,X)$: $X$ is not contained in any hyperplane of $\R^n$ (in other words, its affine hull has dimension $n$) and for all distinct $x_1,x_2\in X$ there are parallel hyperplanes $V_1,V_2\subset \R^n$ such that $V_i$ supports $X$ in $x_i$, for $i=1,2$.
    \item $Q(n,A,X)$: For all $x_1,x_2\in X$, the convex bodies $x_1+A$ and $x_2+A$ touch each other.
    \item $Q^*(n,A,X)$: Same as $Q(n,A,X)$, but we additionally require that $A=-A$ (i.e.\ that $A$ be centrally symmetric).
\end{itemize}
Furthermore, let us set
\begin{align}
    p_n &\coloneqq \sup\left\{ |X|:\, \exists\, \text{$X\subset \R^n$ finite:}\ P(n,X) \right\} ,\\
    q_n &\coloneqq \sup\left\{ |X|:\, \exists\, \text{$X\subset \R^n$ finite, $A\subset \R^n$ convex body:}\ Q(n,X,A) \right\} ,\\
    q_n^* &\coloneqq \sup\left\{ |X|:\, \exists\, \text{$X\subset \R^n$ finite, $A=-A\subset \R^n$ convex body:}\ Q^*(n,X,A) \right\} ,
\end{align}
where $|\cdot|$ denotes the cardinality of a finite set, i.e.\ the number of elements it contains. The geometrically intuitive fact that the $2^n$ vertices of the hypercube satisfy $P(n,X)$ --- and hence $p_n\geq 2^n$ --- has been discussed in Section~\ref{subsec:hypercube}, so we will not dwell on it further. The problem is to show that $p_n\leq 2^n$. The proof can be broken down into the following chain of inequalities:
\begin{equation}
    2^n \leq p_n \textleq{(i)} q_n \texteq{(ii)} q_n^* \textleq{(iii)} 2^n\, .
\end{equation}
We now justify one by one the three crucial steps (i)--(iii):
\begin{enumerate}[i]

    \item In fact, for all $n$ and for all sets $X\subseteq \R^n$ we have that $P(n,X)\Longrightarrow Q\left(n,-\co(X),X\right)$, where $\co$ denotes the convex hull. To see this, assume that $P(n,X)$ holds. Then, for $x_1,x_2\in X$ with $x_1\neq x_2$ there exists a hyperplane $V\subset \R^n$ such that the set $X$, and hence also the convex body $\co(X)$, is entirely contained between $x_1+V$ and $x_2+V$. Multiplying by $-1$ and translating, we see that $x_1-\co(X)$ is entirely contained between $V$ and $x_1-x_2+V$, and analogously $x_2-\co(X)$ is entirely contained between $V$ and $x_2-x_1+V$. Since $x_1\neq x_2$, we see that the convex bodies $x_1-\co(X)$ and $x_2-\co(X)$ are entirely contained into each of the two closed half-spaces determined by $V$. This implies that their interiors, which are instead contained into the corresponding \emph{open} half-spaces, are disjoint. Remembering that $0\in (x_1-\co(X))\cap (x_2-\co(X))$, we see that in fact $x_1-\co(X)$ and $x_2-\co(X)$ touch each other.

    \item We show that for all $n$, for all finite $X\subseteq \R^n$, and for all $A\subseteq \R^n$,
    \begin{equation}
        Q(n,A,X)\quad \Longleftrightarrow\quad Q\left(n,\wtildea{A},X\right) ,
        \label{Q-iff-Q*}
    \end{equation}
    so that naturally $q_n=q_n^*$. Start by noting the following: for a set $A\subseteq \R^n$ and two points $x,y\in \R^n$,
    \begin{equation}
    (x+A)\cap(y+A)\neq \emptyset\quad\Longleftrightarrow\quad \frac{x-y}{2}\in \wtildea{A}\, ,
    \label{fact-of-life}
    \end{equation}
    where $\wtildea{A}$ is the Minkowski symmetrisation of $A$. Therefore, for fixed $x_1,x_2\in X$, we have that $(x_1+A)\cap (x_2+A)\neq \emptyset$ if and only if $\frac{x_1-x_2}{2}\in \wtildea{A}$. Since $A$ and $\wtildea{A}$ have the same Minkowski symmetrisation, this is also equivalent to $\left(x_1+\wtildea{A}\right)\cap \left(x_2+\wtildea{A}\right)\neq \emptyset$. In other words,
    \begin{equation}
        (x_1+A)\cap (x_2+A)\neq \emptyset \quad \Longleftrightarrow\quad \left(x_1+\wtildea{A}\right)\cap \left(x_2+\wtildea{A}\right) \neq \emptyset\, .
    \end{equation}
    Applying this to $\inter A$ instead of $A$, we get that
    \begin{equation}
    \begin{aligned}
        &\inter(x_1+A)\cap \inter(x_2+A) = \left(x_1+\inter(A)\right) \cap \left(x_2+\inter(A)\right) =\emptyset \\
        &\qquad \Longleftrightarrow\quad \left(x_1+\wtildea{\inter A}\right)\cap \left(x_2+\wtildea{\inter A}\right) = \inter\left(x_1+\wtildea{A}\right)\cap \inter\left(x_2+\wtildea{A}\right) = \emptyset\, ,
    \end{aligned}
    \end{equation}
    where the identity $x_i+\wtildea{\inter A} = \inter\left(x_i+\wtildea{A}\right)$ follows from Lemma~\ref{interior Minkowski symm lemma}. We have therefore proved that the convex bodies $x_1+A$ and $x_2+A$: (a)~intersect if and only if so do $x_1+\wtildea{A}$ and $x_2+\wtildea{A}$; and (b)~have disjoint interiors if and only if so do $x_1+\wtildea{A}$ and $x_2+\wtildea{A}$. In other words, $x_1+A$ and $x_2+A$ touch each other if and only if also $x_1+\wtildea{A}$ and $x_2+\wtildea{A}$ touch each other.

    \item We now show that $q_n^*\leq 2^n$. To this end, pick a convex body $A=-A\subset \R^n$ and some set $X\subseteq \R^n$ such that $Q^*(n,A,X)$ holds. Set $B\coloneqq \co(X)$. By~\eqref{fact-of-life}, for all $x_1,x_2\in X$ it must hold that $\frac{x_1-x_2}{2}\in \wtildea{A} = A$. Then we claim that for all $x\in X$,
    \begin{equation}
        \frac{x+B}{2}\subseteq x+A\, .
        \label{x+B-disjoint}
    \end{equation}
    To see this, up to taking the convex hull it suffices to show that $\frac{x+y}{2}\subseteq x+A$ for all $y\in X$. And indeed, thanks to the above observation $\frac{x+y}{2} = x + \frac{y-x}{2}\in x+A$. This proves~\eqref{x+B-disjoint}. As an immediate consequence of this together with $Q^*(n,A,X)$, observe that the interiors of the convex bodies $\frac{x+B}{2}$, indexed by $x\in X$, are all disjoint and moreover contained in $B$, because this is convex. \tcb{Since convex bodies are well known to be Lebesgue measurable~\cite{Lang1986}, we can now deduce that the volume of $B$ is at least equal to the sum of the volumes of the bodies $\frac{x+B}{2}$, in formula} 
    \begin{equation}
        \vol(B) \geq \sum_{x\in X} \vol\left(\frac{x+B}{2}\right) = \sum_{x\in X} \frac{1}{2^n} \vol\left(x+B\right) = \sum_{x\in X} \frac{1}{2^n} \vol\left(B\right) = \frac{|X|}{2^n}\, \vol(B)\, .
    \end{equation}
    Since $\vol(B)>0$ because $B$ is a convex body, we obtain that $|X|\leq 2^n$, as claimed.
\end{enumerate}
This concludes the proof.
\end{proof}

\section{Prism Theories}\label{sec:PrismTheories}
Expanding a GPT to higher dimensions is a way to explore systems with a variable number of degrees of freedom, but which are governed by a consistent set of relationships. For example, the state space of an $n$-sided ordinary, classical die is represented in GPT form by a simplex with $n$ vertices (representing deterministic preparations of a particular outcome); this allows us to accommodate systems having many degrees of freedom by generalising the same basic geometric pattern to higher dimensions.
Though this works in a straightforward way for classical theory, the situation in quantum theory is more nuanced. The state space of the qubit is represented by the Bloch sphere in three dimensions, but the state space of a qutrit possesses a complicated geometry~\cite{Tabia2013} which is not simply given by a sphere in higher dimensions. Theories using hyperspheres of higher dimension, so-called $D$-balls, are discussed in~\cite{Masanes2014, Krumm2019}; in these the authors aim to isolate the 3-sphere as a the necessary state space for quantum theory based on physical requirements.

Here, we introduce a method for expanding given geometries to higher dimensions in a generic way. We do this by taking the Cartesian product of shapes in lower dimensional spaces. In Figure~\ref{fig:simplexPrisms} we visualise some state spaces shaped as simplices and their corresponding effects, as well as a simplex prism $\mathcal{S}_3\times\mathcal{S}_3$.

However, this way of incorporating new degrees of freedom, although mathematically consistent, does not have a direct operational interpretation: new variables do not have to be independent of the old ones.

It is a feature of the Cartesian product that the product of any two convex sets will produce a new convex set. We can use this feature as the basis to construct new, higher dimensional GPTs by taking the Cartesian product (denoted $\times$) of lower dimensional state spaces. {The resulting GPT is called a \emph{prism theory}. We give a formal definition below:}


\begin{Def}[Prism theories] \label{def:prism-theories}
Let $A=(V_A,C_A,u_A)$ and $B=(V_B,C_B,u_B)$ be two GPTs. The \textbf{prism theory} $A\oplus B = \left(V_{A\oplus B}, C_{A\oplus B}, u_{A\oplus B}\right)$ is defined as follows:
\begin{enumerate}[i]
    \item $V_{A\oplus B} \coloneqq \ker\left( (u_A, 0) - (0,u_B)\right)\subset V_A\oplus V_B$ is the subspace of $V_A\oplus V_B$ given by the kernel of the functional $(u_A, 0) - (0,u_B)$ whose action is defined by $\left((u_A, 0) - (0,u_B)\right)(x,y) \coloneqq u_A(x) - u_B(y)$;
    \item $C_{A\oplus B}\coloneqq \left(C_A\oplus C_B\right)\cap V_{A\oplus B}$;
    \item $u_{A\oplus B}$ is the restriction of $(u_A,0)$ (equivalently, of $(0,u_B)$) to $V_{A\oplus B}$.
\end{enumerate}
\end{Def}

\begin{figure}
     \centering
     \begin{subfigure}[b]{0.31\textwidth}
         \centering
         \includegraphics[width=\textwidth]{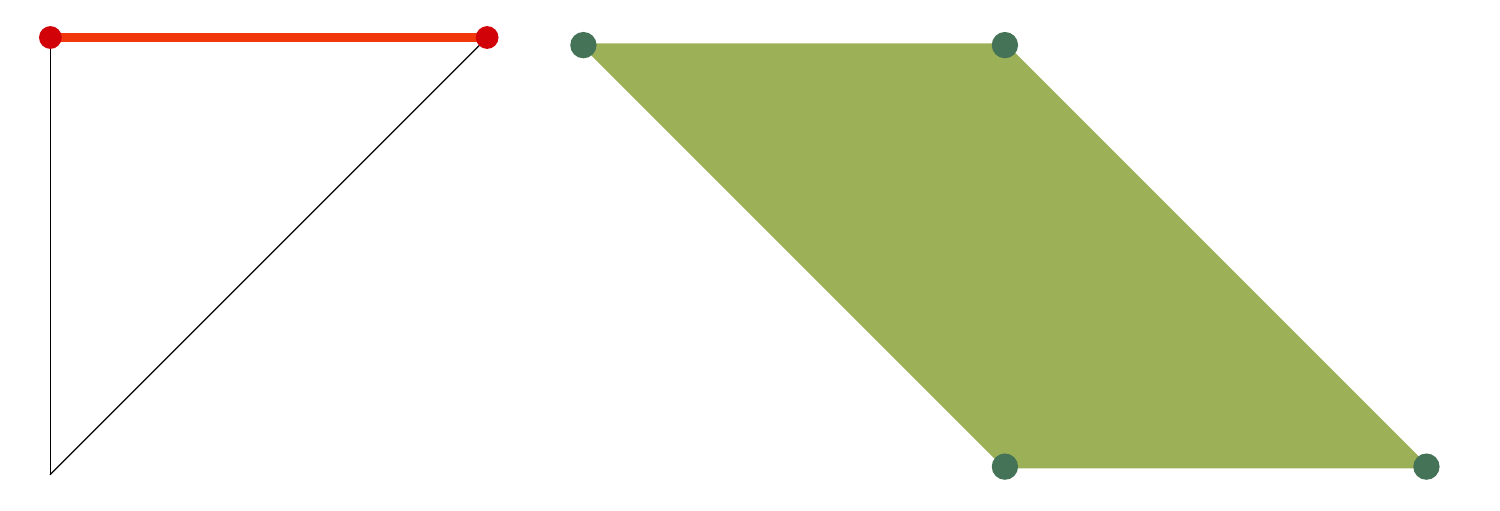}
         \caption{States and effects for simplex with 2 pure states}
         \label{fig:Simplex2dState}
     \end{subfigure}
     \hfill
     \begin{subfigure}[b]{0.31\textwidth}
         \centering
         \includegraphics[width=\textwidth]{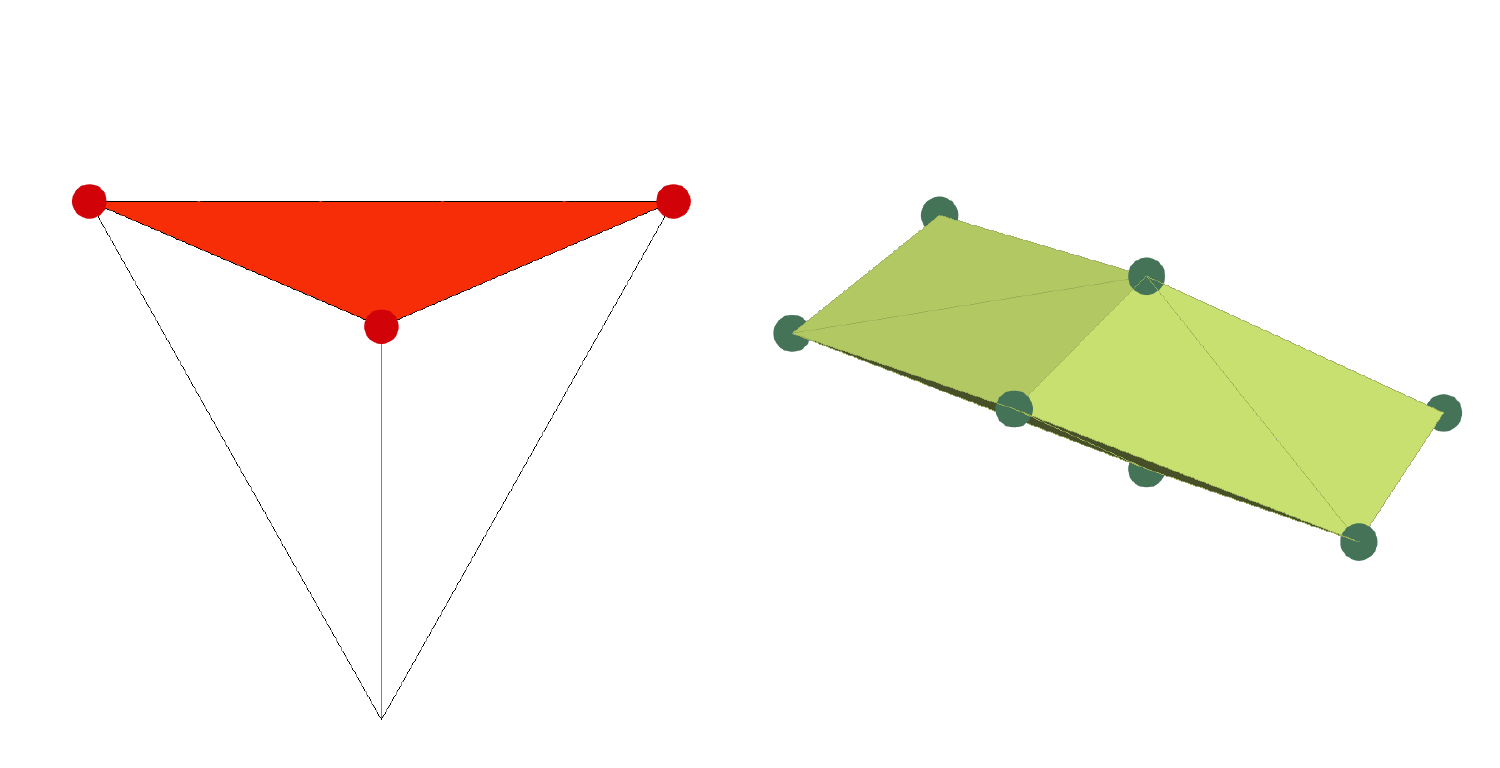}
         \caption{States and effects for simplex with 3 pure states}
         \label{fig:Simplex3dState}
     \end{subfigure}
     \hfill
     \begin{subfigure}[b]{0.31\textwidth}
         \centering
         \includegraphics[width=\textwidth]{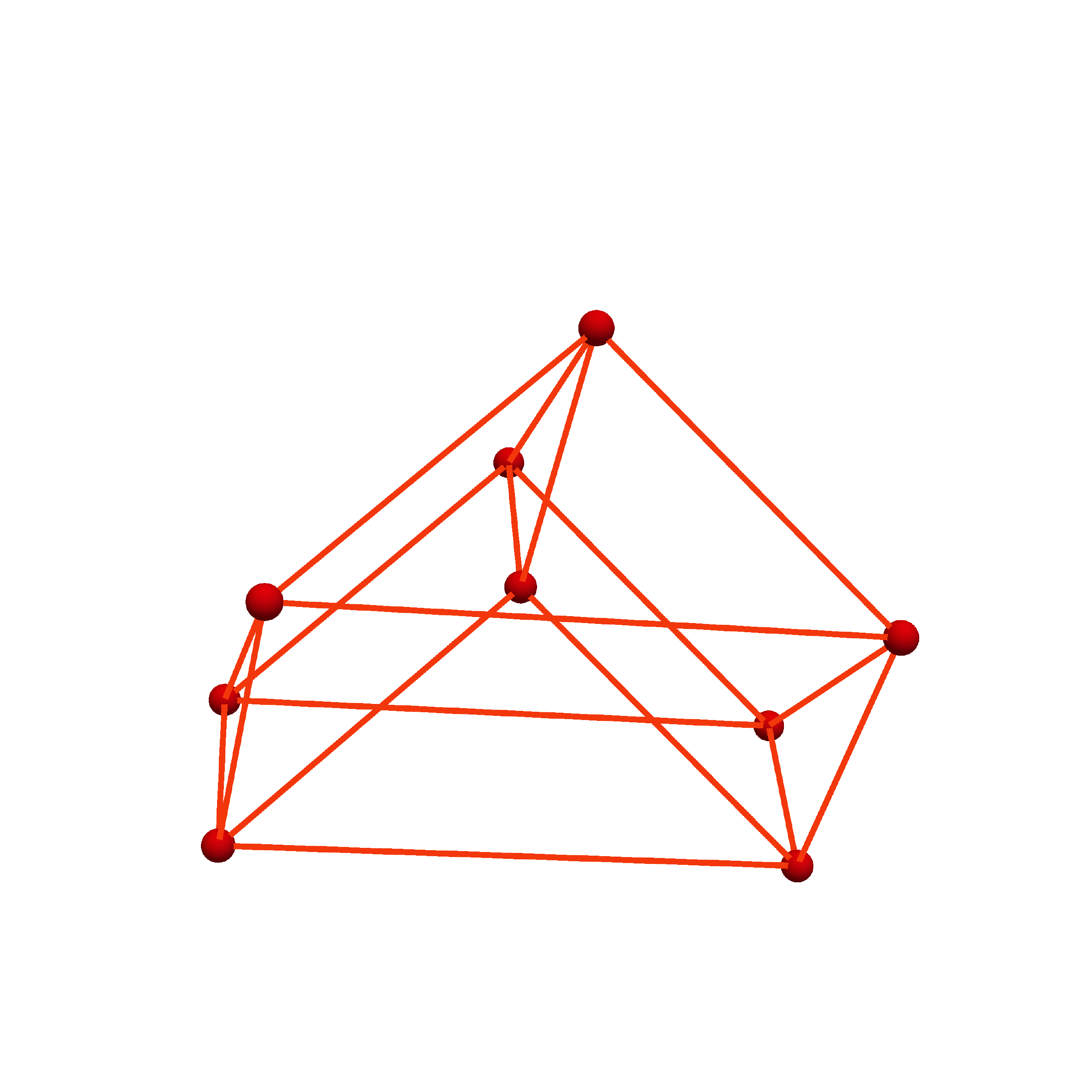}
         \caption{Projection of the shape formed by $\mathcal{S}_3\times\mathcal{S}_3.$ }
         \label{fig:SimplexPrism33}
     \end{subfigure}
     \hfill
     \caption{Depiction of GPTs based on simplices. Image (\subref{fig:Simplex2dState}) shows the states in a theory based on $\mathcal{S}_2$. The two pink points correspond to the extremal states of the theory, and the thin black lines connect these to the origin. The green square shows the space of possible effects, with the darker points signifying the extremal effects. In (\subref{fig:Simplex3dState}) we see the theory corresponding to $\mathcal{S}_3$, in which we've added a dimension. On the right, in (\subref{fig:SimplexPrism33}) we see a representation of the product of two simplices, $\mathcal{S}_3\times \mathcal{S}_3$. Since this yields a shape in four dimensions, we have used a projection to visualise it in three dimensions. Here the entire volume of the polytope should be understood as representing the space of mixed states, with the pink points again corresponding to pure states. } \label{fig:simplexPrisms}
\end{figure}

To unpack the above somewhat complicated definition, it is useful to look at the state spaces. Since the host vector space $V_{A\oplus B}$ is a subspace of the simple direct sum $V_A\oplus V_B$, any state of $A\oplus B$ can also be seen as a vector of the form $\omega_{A\oplus B} = (x,y)\in V_A\oplus V_B$. We observe that item~(ii) implies that in fact $x\in C_A$ and $y\in C_B$, so that $x=\lambda \omega_A$ and $y=\mu \omega_B$, for $\lambda,\mu\geq 0$ and $\omega_A\in \Omega_A \coloneqq C_A\cap u_A^{-1}(1)$, $\omega_B\in \Omega_B \coloneqq C_B\cap u_B^{-1}(1)$. Now, since $\omega_{A\oplus B}$ must belong to the kernel of $(u_A, 0) - (0,u_B)$, we also see that $\lambda=\mu$; if it is a normalised state, then by~(iii) we have that $1=u_{A\oplus B}\left(\omega_{A\oplus B}\right)=(u_A,0)(\lambda \omega_A, \lambda\omega_B) = \lambda$. Therefore, $\omega_{A\oplus B}$ can be simply identified with the pair of states $(\omega_A, \omega_B)$, and vice versa any such pair constitutes a state of $A\oplus B$. We have thus proved the following, which amounts to an intuitive description of the rather cumbersome Definition~\ref{def:prism-theories}:

\begin{lemma} \label{lemma:Cartesian-product}
For any two GPTs $A,B$ with state spaces $\Omega_A, \Omega_B$, the state space of the prism theory $A\oplus B$ is simply the Cartesian product of $\Omega_A$ and $\Omega_B$. In formula,
\begin{equation}
\Omega_{A\oplus B} = \Omega_A\times \Omega_B\, .
\end{equation}
\end{lemma}

\begin{rem}
If $A,B$ are two GPTs with dimensions $\dim A=d_A$ and $\dim B=d_B$, thanks to Lemma~\ref{lemma:Cartesian-product} we have that
\begin{equation}
    \dim\left(A\oplus B\right) = \left(d_A-1\right)\left(d_B-1\right)+1\, .
\end{equation}
\end{rem}

\section{On measurement normalisation} \label{app:restricting_search}


Here we construct an example of a GPT with state space $\Omega$ and effect space $E$ in which one can find three states $(\omega_i)_{i=1,2,3}\subset \Omega$ and three extremal effects $(e_j)_{j=1,2,3}\subset E$ satisfying $e_j \cdot \omega_i = \delta_{ij}$, but such that $(\omega_i)_{i=1,2,3}$ are not perfectly distinguishable, i.e.\ there does not exist a \emph{measurement} $(f_k)_{k=1,2,3}$ such that $f_k\cdot \omega_i = \delta_{ik}$. The reason why this is possible, naturally, is that only collections of effects $(f_k)_k$ satisfying $\sum_k f_k = u$, with $u$ being the order unit, can represent physical measurements.

The state space of the GPT we have in mind is --- once again! --- shaped as a $3$-dimensional cube. More precisely, we consider the $n=3$ case of the GPT constructed in Section~\ref{subsec:hypercube} (see in particular~\eqref{cubic_cone} there). Its state space is depicted in Figure~\ref{cube_fig}. We identify there three states $\omega_1,\omega_2,\omega_3$, with coordinates
\bb
\omega_1 \coloneqq \left( 1,1,1,1 \right)^\intercal ,\quad \omega_2 \coloneqq \left( 1,-1,1,-1 \right)^\intercal ,\quad \omega_3 \coloneqq \left( 1,-1,-1,1 \right)^\intercal ,
\label{omega_123}
\ee
and five auxiliary states $\rho_0,\rho_1$ and $\sigma_1,\sigma_2,\sigma_3$, defined by
\bb
\rho_0 &\coloneqq \left( 1,-1,1,1 \right)^\intercal ,\qquad \rho_1 \coloneqq \left( 1,1,-1,-1 \right)^\intercal ,
\label{rho_01}
\ee
\bb
\sigma_1 \coloneqq \left( 1,-1,-1,-1 \right)^\intercal ,\quad  \sigma_2 \coloneqq \left( 1,1,-1,1 \right)^\intercal ,\quad  \sigma_3 \coloneqq \left( 1,1,1,-1 \right)^\intercal .
\label{sigma_123}
\ee
Note that the first coordinate represents the normalisation, in accordance with the notation of~\eqref{cubic_cone}, and the last three identify the position of the state in the $3$-dimensional `section' space depicted in Figure~\ref{cube_fig}.

We now construct the three extremal effects $(e_j)_{j=1,2,3}\subset E$ satisfying $e_j \cdot \omega_i = \delta_{ij}$. In the dual space set
\bb
e_1 \coloneqq \frac12 \left(1,1,0,0\right) ,\quad  e_2 \coloneqq \frac12 \left(1,0,0,-1\right) ,\quad  e_3 \coloneqq \frac12 \left(1,0,-1,0\right) .
\label{e_123}
\ee
(Note that the states were represented by column vectors, so the effects are represented by row vectors.) Note that indeed $e_j \cdot \omega_i = \delta_{ij}$. Moreover, since a generic effect is of the form $(c,y_1,y_2,y_3)$, with $\min\{c,1-c\}\geq \sum_i |y_i|$, it follows that each $e_j$ is an extremal effect. The faces of the state space on which $e_1=0$, $e_2=0$, and $e_3=0$ are depicted in Figure~\ref{cube_fig} as coloured in red, blue, and green, respectively.

We now show that the states in~\eqref{omega_123} are not perfectly distinguishable. A first clue that this may be the case can be obtained by noting that the three effects in~\eqref{e_123} satisfy $\sum_i e_i = \frac12 \left(3, 1,-1,-1\right) \not\leq \left(1,0,0,0\right)=u$, where $\not\leq$ signifies that the inequality $\leq$ can be violated if both sides are evaluated on certain states in $\Omega$. This means that the collection $(e_1,e_2,e_3)$ does not constitute a measurement. ITo turn this observation into a fully-fledged proof, one observes that the three effects in~\eqref{e_123} are the \emph{only} ones that can satisfy $e_j \cdot \omega_i = \delta_{ij}$: since they do not form a measurement, the states in~\eqref{omega_123} cannot be perfectly distinguishable.

\begin{figure} \centering
\begin{tikzpicture}
\coordinate (O) at (0,0,0);
\coordinate (A) at (0,\Width,0);
\coordinate (B) at (0,\Width,\Height);
\coordinate (C) at (0,0,\Height);
\coordinate (D) at (\Depth,0,0);
\coordinate (E) at (\Depth,\Width,0);
\coordinate (F) at (\Depth,\Width,\Height);
\coordinate (G) at (\Depth,0,\Height);

\draw[black,fill=gray!30] (O) -- (C) -- (G) -- (D) -- cycle;
\draw[black,fill=red!30] (O) -- (A) -- (E) -- (D) -- cycle;
\draw[black,fill=gray!30] (O) -- (A) -- (B) -- (C) -- cycle;
\draw[black,fill=green!30,opacity=0.7] (D) -- (E) -- (F) -- (G) -- cycle;
\draw[black,fill=gray!30,opacity=0.3] (C) -- (B) -- (F) -- (G) -- cycle;
\draw[black,fill=blue!30,opacity=0.7] (A) -- (B) -- (F) -- (E) -- cycle;

\draw[black] (A) node[above] {$\omega_3$} node {$\bullet$};
\draw[black] (D) node[right] {$\omega_2$} node {$\bullet$};
\draw[black] (F) node[anchor=south east] {$\omega_1$} node {$\bullet$};
\draw[black] (C) node[anchor=north east] {$\rho_1$} node {$\bullet$};
\draw[black] (E) node[anchor=south west] {$\rho_0$} node {$\bullet$};

\draw[black] (O) node[left] {$\sigma_1$} node {$\bullet$};
\draw[black] (B) node[left] {$\sigma_2$} node {$\bullet$};
\draw[black] (G) node[anchor=north west] {$\sigma_3$} node {$\bullet$};
\end{tikzpicture}
\caption{{A pictorial representation of the construction in~\ref{cube_ex}. The three coloured faces represent the set of states for which $e_1=0$ (red), $e_2=0$ (blue), and $e_3=0$ (green).}}
\label{cube_fig}
\end{figure}
\label{cube_ex}

We will however follow a different reasoning, which has the advantage of providing some quantitative insights. To this end, we will employ the auxiliary states in~\eqref{rho_01} and~\eqref{sigma_123}. We start by noticing that for all $k=1,2,3$ it holds that
$\rho_1 = 2(\sigma_k +\omega_k) - \sum_i \omega_i$.
Now, assume by contradiction that we have found a measurement $(f_k)_k$ satisfying both $\sum_k f_k = u$ and $f_k\cdot \omega_i = \delta_{ik}$. Then
\bb
1 = u(\rho_1) &= \sum_k f_k(\rho_1) = \sum_k f_k\left( 2(\sigma_k +\omega_k) - \sumno_i \omega_i \right) \\
&\geq \sum_k \left( 2 - \sum_i \delta_{ik} \right) = \sum_k (2-1) = 3\, ,
\ee
and we have reached a contradiction.

\end{document}